\documentclass[10pt,a4paper]{article}
\usepackage[utf8]{inputenc}
\usepackage[english]{babel}
\usepackage[T1]{fontenc}
\usepackage{xspace}
\usepackage[colorlinks=true,citecolor=black,linkcolor=black,urlcolor=black]{hyperref}	
\usepackage{amssymb,amsmath,amsthm}
\usepackage{tikz}
\usepackage{caption}
\usepackage{subcaption}
\usepackage[hmargin=3cm,vmargin=3cm]{geometry}
\usetikzlibrary{calc,snakes}
\usetikzlibrary{decorations.pathmorphing}
\usepackage{todonotes}
\tikzset{
  bigblue/.style={circle, draw=blue!80,fill=blue!40,thick, inner sep=1.5pt, minimum size=5mm},
  bigred/.style={circle, draw=red!80,fill=red!40,thick, inner sep=1.5pt, minimum size=5mm},
  bigblack/.style={circle, draw=black!100,fill=black!40,thick, inner sep=1.5pt, minimum size=5mm},
  bluevertex/.style={circle, draw=blue!100,fill=blue!100,thick, inner sep=0pt, minimum size=2mm},
  redvertex/.style={circle, draw=red!100,fill=red!100,thick, inner sep=0pt, minimum size=2mm},
  blackvertex/.style={circle, draw=black!100,fill=black!100,thick, inner sep=0pt, minimum size=2mm},  
  whitevertex/.style={circle, draw=black!100,fill=white!100,thick, inner sep=0pt, minimum size=2mm},  
  smallblack/.style={circle, draw=black!100,fill=black!100,thick, inner sep=0pt, minimum size=1mm},  
}

\newtheorem{theorem}{Theorem}
\newtheorem{proposition}[theorem]{Proposition}
\newtheorem{lemma}[theorem]{Lemma}

\newtheorem{corollary}[theorem]{Corollary}

\newcommand{\ecto}{\ensuremath{\stackrel{ec}{\longrightarrow}}}
\newcommand{\SWITCH}[1]{\textsc{SW-#1-Colouring}\xspace}
\newcommand{\VDEL}[1]{\textsc{VD-#1-Colouring}\xspace}
\newcommand{\EDEL}[1]{\textsc{ED-#1-Colouring}\xspace}
\newcommand{\Ptime}{\textsf{P}\xspace}
\newcommand{\W}{\textsf{W[1]}\xspace}
\newcommand{\NP}{\textsf{NP}\xspace}
\newcommand{\XP}{\textsf{XP}\xspace}

\newcommand{\set}[1]{\left\{#1\right\}}
\newcommand{\abs}[1]{\left|#1\right|}

\newlength{\atextwidth}
\newcommand{\problemparam}[4]{
  \vspace{1mm}
\noindent\fbox{
  \begin{minipage}{\atextwidth}
  \begin{tabular*}{\textwidth}{@{\extracolsep{\fill}}lr} #1 & {\bf{Parameter:}} #3 \\ \end{tabular*}
  {\bf{Input:}} #2  \\
  {\bf{Question:}} #4
  \end{minipage}
  }
  \vspace{1mm}
}

\newcommand{\specialcell}[2][c]{%
  \begin{tabular}[#1]{@{}l@{}}#2\end{tabular}}

\newcommand{\footremember}[2]{%
    \footnote{#2}
    \newcounter{#1}
    \setcounter{#1}{\value{footnote}}%
}
\newcommand{\footrecall}[1]{%
    \footnotemark[\value{#1}]%
}

\title{Graph modification for edge-coloured and signed graph homomorphism problems: parameterized and classical complexity\thanks{This research was financed by the ANR project HOSIGRA (ANR-17-CE40-0022) and the IFCAM project ``Applications of graph homomorphisms'' (MA/IFCAM/18/39). A shorter version of this paper appeared in the proceedings of the conference IPEC'19~\cite{IPEC}.}}

\author{%
  Florent Foucaud\footremember{limos}{LIMOS, CNRS UMR 6158, Universit\'e Clermont Auvergne, Aubi\`ere, France.}\footremember{labri}{Univ. Bordeaux, Bordeaux INP, CNRS, LaBRI, UMR5800, F-33400 Talence, France.}\footremember{orleans}{Univ. Orléans, INSA Centre Val de Loire, LIFO EA 4022, F-45067 Orléans Cedex 2, France.}
  \and Herv\'e Hocquard\footrecall{labri}
  \and Dimitri Lajou\footrecall{labri}
  \and Valia Mitsou\footremember{irif}{Université Paris-Diderot, IRIF, CNRS, 75205, Paris, France.}
  \and Th\'eo Pierron\footremember{liris}{Univ Lyon, Universit\'e Claude Bernard, CNRS, LIRIS - UMR 5205, F69622, France.}\footrecall{labri}~\footremember{cz}{Faculty of Informatics, Masaryk University, Botanick\'a 68A, 602 00 Brno, Czech Republic.}
  }

\begin{document}

\maketitle

\begin{abstract}
  We study the complexity of graph modification problems with respect to homomorphism-based colouring properties of edge-coloured graphs. A homomorphism from an edge-coloured graph $G$ to an edge-coloured graph $H$ is a vertex-mapping from $G$ to $H$ that preserves adjacencies and edge-colours. We consider the property of having a homomorphism to a fixed edge-coloured graph $H$, which generalises the classic vertex-colourability property. The question we are interested in is the following: given an edge-coloured graph $G$, can we perform $k$ graph operations so that the resulting graph admits a homomorphism to $H$? The operations we consider are vertex-deletion, edge-deletion and switching (an operation that permutes the colours of the edges incident to a given vertex). Switching plays an important role in the theory of signed graphs, that are $2$-edge-coloured graphs whose colours are the signs $+$ and~$-$. We denote the corresponding problems (parameterized by $k$) by \VDEL{$H$}, \EDEL{$H$} and \SWITCH{$H$}. These problems generalise the extensively studied \textsc{$H$-Colouring} problem (where one has to decide if an input graph admits a homomorphism to a fixed target $H$). For $2$-edge-coloured $H$, it is known that \textsc{$H$-Colouring} already captures the complexity of all fixed-target Constraint Satisfaction Problems.

  Our main focus is on the case where $H$ is an edge-coloured graph with at most two vertices, a case that is already interesting since it includes standard problems such as \textsc{Vertex Cover}, \textsc{Odd Cycle Transversal} and \textsc{Edge Bipartization}. For such a graph $H$, we give a \Ptime/\NP-complete complexity dichotomy for all three \VDEL{$H$}, \EDEL{$H$} and \SWITCH{$H$} problems. Then, we address their parameterized complexity. We show that all \VDEL{$H$} and \EDEL{$H$} problems for such $H$ are FPT. This is in contrast with the fact that already for some $H$ of order~$3$, unless \Ptime = \NP, none of the three considered problems is in \XP, since \textsc{$3$-Colouring} is \NP-complete. We show that the situation is different for \SWITCH{$H$}: there are three $2$-edge-coloured graphs $H$ of order~$2$ for which \SWITCH{$H$} is \W-hard, and assuming the ETH, admits no algorithm in time $f(k)n^{o(k)}$ for inputs of size $n$ and for any computable function $f$. For the other cases, \SWITCH{$H$} is FPT.
\end{abstract}

\section{Introduction}
Graph colouring problems such as \textsc{$k$-Colouring} are among the most fundamental problems in algorithmic graph theory. The problem of \textsc{$H$-Colouring} is a homomorphism-based generalisation of \textsc{$k$-Colouring} that is extensively studied~\cite{B18,FV98,HN90,Marx}. Considering a fixed graph $H$, in \textsc{$H$-Colouring} one asks whether an input graph $G$ admits a homomorphism (an edge-preserving vertex-mapping) to $H$. Observe that \textsc{$k$-Colouring} is the same problem as \textsc{$K_k$-Colouring}, where $K_k$ is the complete graph of order $k$ (the order of a graph is its number of vertices).

We will consider parameterized variants of \textsc{$H$-Colouring} where $H$ is an edge-coloured graph. We say that a graph is \emph{$t$-edge-coloured} if its edges are coloured with at most $t$ colours. In this paper, all (edge-coloured) graphs may have loops and multiple edges, but multiple edges of the same colour are irrelevant. (Thus by \emph{graph} we effectively mean \emph{multigraph}.) We sometimes give actual colour names to the colours: red, blue, green. For $2$-edge-coloured graphs, we will use red and blue as the two edge colours. A standard uncoloured graph can be seen as $1$-edge-coloured. For two edge-coloured graphs $G$ and $H$, a {\em homomorphism} from $G$ to $H$ is a vertex-mapping $\varphi:V(G)\to V(H)$ such that, if $xy$ is an edge of colour~$i$ in $G$, then $\varphi(x)\varphi(y)$ is an edge of colour~$i$ in $H$. Whenever such a $\varphi$ exists, we say that $G$ maps to $H$, and we write $G\ecto H$. 

The \textsc{$H$-Colouring} problems are well-studied, see for example~\cite{BBM05,Bthesis,B94,BDHQ05,SignedDicho}. They are special cases of \emph{Constraint Satisfaction Problems} (CSPs). A large set of CSPs can be modeled by homomorphisms of general relational structures to a fixed relational structure $H$~\cite{FV98}. 
The corresponding decision problem is noted as \textsc{$H$-CSP}. When $H$ has only binary relations, $H$ can be seen as an edge-coloured graph (a relation corresponds to the set of edges of a given colour) and \textsc{$H$-CSP} is exactly \textsc{$H$-Colouring}. The complexity of \textsc{$H$-CSP} has been the subject of intensive research in the last decades, since Feder and Vardi conjectured in~\cite{FV98} that \textsc{$H$-CSP} is either in \Ptime or is \NP-complete --- a statement that became known as the Dichotomy Conjecture. The latter conjecture was solved in 2017 in~\cite{B17,Z20} independently; the criterion for \textsc{$H$-CSP} to be in \Ptime is based on certain algebraic properties of $H$. Nevertheless, determining whether a structure $H$ satisfies this criterion is not an easy task (even for targets as simple as oriented trees~\cite{B18}). Thus, the study of more simple and elegant complexity classifications for relevant special cases is of high importance.

The complexity of \textsc{$H$-Colouring} when $H$ is uncoloured is well-understood: it is in \Ptime if $H$ contains a loop or is bipartite; otherwise it is \NP-complete~\cite{HN90}. This was one of the early dichotomy results in the area. On the other hand, when $H$ is a $2$-edge-coloured graph, it was proved that the class of \textsc{$H$-Colouring} problems captures the difficulty of the whole class of \textsc{$H$-CSP} problems~\cite{SignedDicho}, and thus the dichotomy classification for this class of problems is expected to be much more intricate.

Our goal is to study generalisations of \textsc{$H$-Colouring} problems for edge-coloured graphs by enhancing them as \emph{modification problems}. In this setting, given a graph property $\mathcal P$ and a graph operation $\pi$, the graph modification problem for $\mathcal P$ and $\pi$ asks whether an input graph $G$ can be made to satisfy property $\mathcal P$ after applying operation $\pi$ a given number $k$ of times. This is a classic setting studied extensively both in the realms of classical and parameterized complexity, see for example~\cite{Cai96,survey-edge-mod,KR02,LY80,Y81}. In this context, the most studied graph operations are vertex-deletion  and edge-deletion, see the seminal papers~\cite{LY80,Y81}. 

For a fixed graph $H$, let $\mathcal P(H)$ denote the property of admitting a homomorphism to $H$. Certain standard computational problems can be stated as graph modification problems to $\mathcal P(H)$. For example, \textsc{Vertex Cover} is the graph modification problem for property $\mathcal P(K_1)$ and operation vertex-deletion. Similarly, \textsc{Odd Cycle Transversal} and \textsc{Edge Bipartization} are the graph modification problems for $\mathcal P(K_2)$ and vertex-deletion, and $\mathcal P(K_2)$ and edge-deletion, respectively.

When considering edge-coloured graphs with only two edge-colours, another operation of interest is \emph{switching}: to switch at a vertex $v$ is to change the colour of all edges incident with $v$. (Note that a loop does not change its colour under switching.) This operation is of prime importance in the context of signed graphs. A \emph{signed graph} is a $2$-edge-coloured graph in which the two colours are denoted by signs ($+$ and $-$). 
A graph is called \emph{balanced} if it can be switched to be all-positive. The concepts of signed graphs, balance and switching, were introduced and developed in~\cite{H53,Z82} and have many interesting applications, in particular in social networks and biological dynamical systems (see~\cite{HBN10} and the references therein).

The switching operation plays an important role in the study of homomorphisms of signed graphs, a concept defined in~\cite{NRS15} which has many connections to deep questions in structural graph theory. In their definition, before mapping the vertices, one may perform any number of switchings. (Note that when switching at a set $S$ of vertices of a signed graph $G$, the order does not matter: ultimately, only the edges between $S$ and its complement $V(G)\setminus S$ change their sign.) The algorithmic complexity of this problem was studied in~\cite{SignedDicho,BS18,planar,FN14}. Herein, we will consider edge-coloured graph modification problems for property $\mathcal P(H)$ (for fixed edge-coloured graphs $H$) and for graph operations vertex-deletion, edge-deletion and switching.

A parameterized problem is a decision problem where in addition to the input, a parameter is also considered (the parameter is an integer function of the input). Such a problem is \emph{fixed parameter tractable} (FPT) if for any input $I$ with parameter value $k$, it can be solved in time $O(f(k)|I|^{c})$ for a computable function $f$ and integer $c$. It is in the class \XP if it can be solved in time $O(|I|^{g(k)})$ for a computable function $g$. It is \W-hard if all problems in the class \W can be reduced in FPT time to it. For more details, see the books~\cite{BookParamAlgo,DF13}. 

Let us now formally define the problems of interest to us, where $H$ is a fixed edge-coloured graph (the parameter is always $k$).

\problemparam{\VDEL{$H$}}{An edge-coloured graph $G$, an integer $k$.}{$k$.}{Is there a set $S$ of at most $k$ vertices of $G$ such that $(G-S)\ecto H$?}

\problemparam{\EDEL{$H$}}{An edge-coloured graph $G$, an integer $k$.}{$k$.}{Is there a set $S$ of at most $k$ edges of $G$ such that $(G-S)\ecto H$?}

\problemparam{\SWITCH{$H$}}{A $2$-edge-coloured graph $G$, an integer $k$.}{$k$.}{Is there a set $S$ of $k$ vertices of $G$ such that the $2$-edge-coloured graph $G'$ obtained from $G$ by switching at every vertex of $S$ satisfies $G'\ecto H$?}

In the study of the three above problems, one may assume that $H$ is a \emph{core} (that is, $H$ does not have a homomorphism to a proper subgraph of itself). Indeed, it is well-known that for any subgraph $H'$ of $H$ with $H\ecto H'$, we have $G\ecto H$ if and only if $G\ecto H'$~\cite{B94}.

Of course, whenever \textsc{$H$-Colouring} is \NP-complete, all three above problems are \NP-complete, even when $k=0$, and so they are not in \XP (unless $\Ptime=\NP$). This is for example the case when $H$ is a monochromatic triangle: then this is the problem \textsc{$3$-Colouring}. Thus, from the point of view of parameterized complexity, it is of primary interest to consider these problems for edge-coloured graphs $H$ such that \textsc{$H$-Colouring} is in \Ptime. In that case a simple brute-force algorithm iterating over all $k$-subsets of vertices of $G$ implies that the three problems are in \XP and hence the interesting question is whether these problems are FPT or not. For undirected graphs, the only cores $H$ for which \textsc{$H$-Colouring} is in \Ptime are the three connected graphs with at most one edge~\cite{H53} (a single vertex with no edge, a single vertex with a loop, two vertices joined by an edge), so in that case the interest of these problems is limited. However, for many interesting families of edge-coloured graphs $H$, the problem \textsc{$H$-Colouring} is in \Ptime, and the class of such graphs $H$ is not very well-understood, see~\cite{Bthesis,B94,BDHQ05}. Even when $H$ is a $2$-edge-coloured cycle, tree or complete graph, there are infinitely many $H$ with \textsc{$H$-Colouring} \NP-complete and infinitely many $H$ where it is in \Ptime~\cite{Bthesis}.

Recall that when $H$ is a single vertex with no loop, \VDEL{$H$} is exactly \textsc{Vertex Cover}. If $H$ has a single edge, \VDEL{$H$} and \EDEL{$H$} are \textsc{Odd Cycle Transversal} and \textsc{Edge Bipartization}, respectively. For $H$ consisting of a single (blue) loop, \SWITCH{$H$} for $k=|V(G)|$ consists in checking whether the given $2$-edge-coloured graph $G$ is balanced (a problem that is in \Ptime~\cite{SignedDicho}). More generally, \SWITCH{$H$} for $2$-edge-coloured graphs $H$ and $k=|V(G)|$ (that is, the number of switchings is unrestricted) is exactly the problem \textsc{Signed $H$-Colouring} studied in~\cite{SignedDicho,BS18,FN14}.

\paragraph{Related work.} Several works address the parameterized complexity of graph colouring problems. Graph colouring problems parameterized by structural parameters are considered in~\cite{JJ17}. In~\cite{Chitnis2017}, the vertex-deletion variant of \textsc{$H$-List-Colouring} is studied. Graph modification problems for \textsc{Colouring} in specific graph classes and for operations vertex-deletion and edge-deletion are considered, for example in~\cite{Cai} (bipartite graphs, split graphs) and~\cite{TH06} (comparability graphs).

Every problem \VDEL{H} can be encoded as a special \emph{weighted homomorphism} problem \textsc{$H'$-Weighted-Colouring}, as considered in~\cite{weighted}. In that setting, the target $H'$ is a graph with integer weights, and the goal is to find a homomorphism of some input graph $G$ whose weight (i.e. the sum of weights of the images of the vertices of $G$) is at most some given integer $k$. In our setting, we could generalize this problem to edge-coloured graphs and build $H'$ from $H$ by setting all weights to $0$ and adding a new vertex $x$ adjacent to all vertices of $H$ with weight $1$. Now finding a weighted homomorphism of $G$ to $H$ with weight as most $k$ is the same as having a positive solution to \VDEL{$H$} (vertices mapped to $x$ represent the deleted vertices in $S$). A similar notion was studied for general CSPs in~\cite{BM14}. In that setting, only one ``free'' target vertex has weight $0$ and all the others, weight $1$, and the goal is to find a homomorphism of weight at most a given integer $k$. The Boolean CSP version where there are only two target values, $0$ and $1$, and we wish to minimize the number of variables set to $1$, is called the \textsc{Min Ones} problem~\cite{MinOnes}.

Algorithmic problems relative to the operation of \emph{Seidel switching}, similar to our switching, have been considered. Given a (simple) graph $G$, the Seidel switching operation performed at a vertex exchanges all adjacencies and non-adjacencies of $v$. This can be seen as performing a switching operation in a $2$-edge-coloured complete graph, where blue edges are the actual edges of $G$, and red edges are its non-edges. In~\cite{EHHR,Seidel}, the complexity of graph modification problems with respect to the Seidel switching operation and the property of being a member of certain graph classes has been studied. Our work on \SWITCH{$H$} problems can be seen as a variation of these problems, generalised to arbitrary $2$-edge-coloured graphs.

A related switching problem is as follows: given a signed graph $G$ and a positive integer $k$, can it be switched so that there are at most $k$ negative edges? This is shown to be \NP-complete in~\cite{HBN10}.

\paragraph{Our results.} We study the classical and parameterized complexities of the three problems \VDEL{$H$}, \EDEL{$H$} and \SWITCH{$H$}. Our focus is on $t$-edge-coloured graphs $H$ of order at most~$2$ with $t$ an integer ($t=2$ for \SWITCH{$H$}). Despite having just two vertices, \textsc{$H$-Colouring} for such $H$ is interesting and nontrivial; it is proved to be in \Ptime by two different nontrivial methods, see~\cite{BBM05,BDHQ05}. Thus, the three considered problems are in \XP for such $H$. (Recall that for suitable $1$-edge-coloured graphs $H$ of order $1$ or $2$, \VDEL{$H$} and \EDEL{$H$} include \textsc{Vertex Cover} and \textsc{Odd Cycle Transversal}.)

We completely classify the classical complexity of \VDEL{$H$} when $H$ is a $t$-edge-coloured graph of arbitrary order: it is either trivially in \Ptime or \NP-complete. 
It turns out that all \VDEL{$H$} 
problems are FPT when $H$ has order at most $2$. To prove this, we extend a method from~\cite{BDHQ05} and reduce the problem to an FPT variant of \textsc{$2$-Sat}.

%

For \EDEL{$H$}, 
a classical complexity dichotomy seems more difficult to obtain, as there are nontrivial \Ptime cases. We perform such a classification when $H$ is a $t$-edge-coloured graph of order at most~$2$. Similar \textsc{$2$-Sat}-based arguments as for \VDEL{$H$} give a FPT algorithm for \EDEL{$H$} when $H$ has order at most $2$.

For \SWITCH{$H$} when $H$ is a $2$-edge-coloured graph,
the classical dichotomy is again more difficult to obtain. We perform such a classification by using some characteristics of the switch operation and by giving some reductions to well-known \NP-complete problems. In contrast to the two previous cases for the parameterized complexity, we show that for three graphs $H$ of order~$2$, \SWITCH{$H$} is already \W-hard (and cannot be solved in time $f(k)|G|^{o(k)}$ for any computable function $f$, assuming the ETH\footnote{The Exponential Time Hypothesis, ETH, postulates that \textsc{$3$-SAT} cannot be solved in time $2^{o(n)}(n+m)^c$, where $n$ and $m$ are the input's number of variables and clauses, and $c$ is any integer~\cite{IPZ01}.}). For all other $2$-edge-coloured graphs of order~$2$, we prove that \SWITCH{$H$} is FPT.

Table~\ref{table:Results} presents a brief overview of our results, and Table~\ref{table:order2} lists the classical and parameterized complexities of the three considered problems for all $2$-edge-coloured graphs of order at most~$2$.

\begin{table}
\centering
\begin{tabular}{|l|l|l|l|}
\hline
Problem & \VDEL{$H$} & \EDEL{$H$} & \SWITCH{$H$}\\
\hline
\Ptime vs \NP-hard & \specialcell{Dichotomy for\\all graphs (Cor. \ref{cor:VDEL-classification})} & \specialcell{Dichotomy when\\ $|V(H)| \leq 2$ (Thm. \ref{thm:EDEL-order2-poly/NPc})} & \specialcell{Dichotomy when\\ $|V(H)| \leq 2$ (Thm. \ref{thm:SWITCH-order2-poly/NPc})}\\
\hline
\specialcell{FPT vs. \W-hard\\ when $|V(H)| \leq 2$} & All FPT (Thm. \ref{thm:VDEL-order2-FPT}) & All FPT (Thm. \ref{thm:VDEL-order2-FPT}) & Dichotomy (Thms. \ref{thm:SWITCH-order2-FPT}, \ref{thm:W1-general})\\
\hline
\end{tabular}
\caption{Overview of our main results, sorted by problem and by type of classification.}
\label{table:Results}
\end{table}

\begin{table}
\centering
\resizebox{\columnwidth}{!}{%
\begin{tabular}{|c|c|c|c|}
  \hline
  Graph $H$ & \VDEL{$H$} & \EDEL{$H$} & \SWITCH{$H$}\\
  \hline
\scalebox{0.9}{\begin{tikzpicture}[every loop/.style={},scale=1.1]
    \node[blackvertex] (u) at (0,0) {};
    \path[thick,dashed,red] (u)   edge[out=20,in=80,loop, min distance=5mm] node  {} (u);
    \path[thick,blue] (u)   edge[out=100,in=160,loop, min distance=5mm] node  {} (u);
\end{tikzpicture}} & \Ptime & \Ptime &  \Ptime\\
$H^{1}_{rb}$ &&&\\
\hline

\scalebox{0.9}{\begin{tikzpicture}[every loop/.style={},scale=1.1]
    \node[blackvertex] (u) at (0,0) {};
    \path[thick,blue] (u)   edge[out=60,in=120,loop, min distance=5mm] node  {} (u);    
\end{tikzpicture}} & \NP-hard but FPT &  \Ptime &  \Ptime\\
$H^{1}_{b}$ &&&\\
\hline
  
\scalebox{0.9}{\begin{tikzpicture}[every loop/.style={},scale=1.1]    
    \node[blackvertex] (u) at (0,0) {};
    \path (u)+(0,0.2) node {};
  \end{tikzpicture}} & \NP-hard but FPT & \Ptime & \Ptime\\
$H^{1}_-$ &&&\\
\hline

\scalebox{0.9}{\begin{tikzpicture}[every loop/.style={},scale=1.1]
    \node[blackvertex] (u) at (0,0) {};
    \node[blackvertex] (v) at (0.5,0) {};
    \path[thick,red,dashed] (u)   edge[out=60,in=120,loop, min distance=5mm] node  {} (u);
    \path[thick,blue] (v)   edge[out=60,in=120,loop, min distance=5mm] node  {} (v);
   \end{tikzpicture}} & \NP-hard but FPT & \Ptime & \Ptime\\
$H^{2-}_{r,b}$ &&&\\
\hline

\scalebox{0.9}{\begin{tikzpicture}[every loop/.style={},scale=1.1]
    \node[blackvertex] (u) at (0,0) {};
    \node[blackvertex] (v) at (1.5,0) {};
    \path (u)+(0,0.2) node {};
    \draw[thick,blue] (u)--(v);
  \end{tikzpicture}} & \NP-hard but FPT & \NP-hard but FPT & \Ptime\\
$H^{2b}_{-,-}$ &&&\\
\hline

\scalebox{0.9}{\begin{tikzpicture}[every loop/.style={},scale=1.1]
    \node[blackvertex] (u) at (0,0) {};
    \node[blackvertex] (v) at (1.5,0) {};
    \draw[thick,blue] (u)--(v);
    \path[thick,dashed,red] (u)   edge[out=60,in=120,loop, min distance=5mm] node  {} (u);
    \path[thick,blue] (v)   edge[out=60,in=120,loop, min distance=5mm] node  {} (v);
  \end{tikzpicture}} & \NP-hard but FPT & \NP-hard but FPT & \NP-hard but FPT\\
$H^{2b}_{r,b}$ &&&\\
\hline

\scalebox{0.9}{\begin{tikzpicture}[every loop/.style={},scale=1.1]
    \node[blackvertex] (u) at (0,0) {};
    \node[blackvertex] (v) at (1.5,0) {};
    \draw[thick,blue] (u)--(v);
    \path[thick,dashed,red] (u)   edge[out=60,in=120,loop, min distance=5mm] node  {} (u);
\end{tikzpicture}} & \NP-hard but FPT & \NP-hard but FPT & \NP-hard but FPT \\
$H^{2b}_{r,-}$ & & &\\
\hline

\scalebox{0.9}{\begin{tikzpicture}[every loop/.style={},scale=1.1]
    \node[blackvertex] (u) at (0,0) {};
    \node[blackvertex] (v) at (1.5,0) {};
    \draw[thick,blue] (u)--(v);
    \path[thick,dashed,red] (u)   edge[out=60,in=120,loop, min distance=5mm] node  {} (u);
    \path[thick,dashed,red] (v)   edge[out=60,in=120,loop, min distance=5mm] node  {} (v);
\end{tikzpicture}} & \NP-hard but FPT & \NP-hard but FPT & \Ptime\\
$H^{2b}_{r,r}$ & & &\\
\hline

\scalebox{0.9}{\begin{tikzpicture}[every loop/.style={},scale=1.1]
    \node[blackvertex] (u) at (0,0) {};
    \node[blackvertex] (v) at (1.5,0) {};
    \path (u)+(0,0.3) node {};
    \draw[thick,blue] (u) to[bend right=20] (v);
    \draw[thick,dashed,red] (v) to[bend right=20] (u);   
  \end{tikzpicture}} & \NP-hard but FPT & \NP-hard but FPT & \Ptime\\
$H^{2rb}_{-,-}$ &&&\\
\hline

\scalebox{0.9}{\begin{tikzpicture}[every loop/.style={},scale=1.1]
    \node[blackvertex] (u) at (0,0) {};
    \node[blackvertex] (v) at (1.5,0) {};
    \draw[thick,blue] (u) to[bend right=20] (v);
    \draw[thick,dashed,red] (v) to[bend right=20] (u);  
    \path[thick,dashed,red] (u)   edge[out=60,in=120,loop, min distance=5mm] node  {} (u);
    \path[thick,blue] (v)   edge[out=60,in=120,loop, min distance=5mm] node  {} (v);
   \end{tikzpicture}} & \NP-hard but FPT & \NP-hard but FPT & \NP-hard and \W-h \\
$H^{2rb}_{r,b}$ &&&\\
\hline

\scalebox{0.9}{\begin{tikzpicture}[every loop/.style={},scale=1.1]
    \node[blackvertex] (u) at (0,0) {};
    \node[blackvertex] (v) at (1.5,0) {};
    \draw[thick,blue] (u) to[bend right=20] (v);
    \draw[thick,dashed,red] (v) to[bend right=20] (u);  
    \path[thick,dashed,red] (u)   edge[out=60,in=120,loop, min distance=5mm] node  {} (u);
  \end{tikzpicture}} & \NP-hard but FPT & \NP-hard but FPT & \NP-hard and \W-h\\
$H^{2rb}_{r,-}$ &&&\\
\hline

\scalebox{0.9}{\begin{tikzpicture}[every loop/.style={},scale=1.1]
    \node[blackvertex] (u) at (0,0) {};
    \node[blackvertex] (v) at (1.5,0) {};
    \draw[thick,blue] (u) to[bend right=20] (v);
    \draw[thick,dashed,red] (v) to[bend right=20] (u); 
    \path[thick,dashed,red] (u)   edge[out=60,in=120,loop, min distance=5mm] node  {} (u);
    \path[thick,dashed,red] (v)   edge[out=60,in=120,loop, min distance=5mm] node  {} (v);
\end{tikzpicture}} & \NP-hard but FPT & \NP-hard but FPT & \NP-hard and \W-h\\
$H^{2rb}_{r,r}$ &&&\\
\hline

\end{tabular}
}
\caption{Our results for target graphs $H$ of order at most~$2$ (up to inversion of edge-colours, there are twelve such graphs, see Section~\ref{sec:classical}).}
\label{table:order2}
\end{table}

%

Our paper is structured as follows. In Section~\ref{sec:prelim}, we state some definitions and make some preliminary observations in relation with the literature. 
In Section~\ref{sec:classical}, we study the classical complexity of the three considered problems. We address their parameterized complexity in Section~\ref{sec:param}. Finally, we conclude in Section~\ref{sec:conclu}.

\section{Preliminaries and known results}\label{sec:prelim}



\subsection{Some known complexity dichotomies}

Recall that whenever \textsc{$H$-Colouring} is \NP-complete, \VDEL{$H$}, \EDEL{$H$} and \SWITCH{$H$} are \NP-complete (even for $k=0$), and thus are not in \XP, unless \Ptime = \NP. For example, this is the case when $H$ is a monochromatic triangle. When \textsc{Signed $H$-Colouring} (this is \SWITCH{$H$} for $k=|V(G)|$, see~\cite{SignedDicho}) is \NP-complete, then \SWITCH{$H$} is \NP-complete (but could still be in \XP or FPT).

On the other hand, when \textsc{$H$-Colouring} is in \Ptime, all three problems are in \XP for parameter $k$ (by a brute-force algorithm iterating over all $k$-subsets of vertices of $G$, performing the operation on these $k$ vertices, and then solving \textsc{$H$-Colouring}):

\begin{proposition}
Let $H$ be an edge-coloured graph such that \textsc{$H$-Colouring} is in \Ptime. Then, \VDEL{$H$}, \EDEL{$H$} and \SWITCH{$H$} can be solved in time $|G|^{O(k)}$.
\end{proposition}

When $k=0$ and $H$ is $1$-coloured, we have the following classic theorem.

\begin{theorem}[Hell and Ne\v{s}et\v{r}il \cite{HN90}]\label{thm:HNdicho}
Let $H$ be a $1$-edge-coloured graph. \textsc{$H$-Colouring} is in \Ptime if the core of $H$ has at most one edge ($H$ is bipartite or has a loop), and \NP-complete otherwise.
\end{theorem}


There is no analogue of Theorem~\ref{thm:HNdicho} for edge-coloured graphs. In fact, it is proved in~\cite[Section~3]{SignedDicho} that a dichotomy classification for \textsc{$H$-Colouring} restricted to $2$-edge-coloured $H$ would imply a dichotomy for all fixed-target CSP problems. Thus, no simple combinatorial classification is expected to exist. Even for trees, cycles or complete graphs, such classifications are far from trivial, see the PhD thesis~\cite{Bthesis} for an overview of some partial results highlighting the difficulty of the problem. Some classifications exist for certain classes of graphs $H$, such as those of order at most~$2$ (see~\cite[Section 3]{BBM05} and \cite[Section 2.1]{BDHQ05}) or paths~\cite[Sections 2 and 3]{B94}.

For \SWITCH{$H$} with $k=|V(G)|$, (that is, \textsc{Signed $H$-Colouring}), we have the following (where the \emph{switching core} of a $2$-edge-coloured graph is a notion of core where an arbitrary number of switchings can be performed before the self-mapping).

\begin{theorem}[Brewster et al. \cite{SignedDicho,BS18}]\label{thm:dicho-signed-col}
Let $H$ be a signed graph. \textsc{Signed $H$-Colouring} is in \Ptime if the switching core of $H$ has at most two edges, and \NP-complete otherwise.
\end{theorem}

Note that $2$-edge-coloured graphs where the switching core has at most two edges either have one vertex (with zero loop, one loop or two loops of different colours), or two vertices (with either one edge or two parallel edges of different colours joining them)~\cite{SignedDicho}. If there are two vertices joined by one edge and a loop at one of the vertices, we can switch at the non-loop vertex if necessary to obtain one edge-colour, and then retract the whole graph to the loop-vertex, so this is not a core.





\subsection{Homomorphism dualities and FPT time}
\label{sec:duality}

For a $t$-edge-coloured graph $H$, we say that $H$ has the \emph{duality property} if there is a set $\mathcal F(H)$ of $t$-edge-coloured graphs such that, for any $t$-edge-coloured graph $G$, $G\ecto H$ if and only if no graph $F$ of $\mathcal F(H)$ satisfies $F\ecto G$. If $\mathcal F(H)$ is finite, we say that $H$ has the \emph{finite duality property}.
If checking whether any graph $F$ in $\mathcal F(H)$ satisfies $F\ecto G$ (for an input edge-coloured graph $G$) is in \Ptime, we say that $H$ has the \emph{polynomial duality property}. This is in particular the case when $\mathcal F(H)$ is finite. For such $H$, \textsc{$H$-Colouring} is in \Ptime. This topic is explored in detail for edge-coloured graphs in~\cite{BBM05}. By a simple bounded search tree argument, we get the following:

\begin{proposition}
\label{prop:finite-duality-FPT}
Let $H$ be an edge-coloured graph with the finite duality property. Let $c=\max\{|V(F)|, F\in\mathcal F(H)\}$. 
The problems \VDEL{$H$} and \SWITCH{$H$} can be solved in time $O(f(\mathcal F(H))n^{c})$ for some computable function $f$.
The problem \EDEL{$H$}can be solved in time $O(f(\mathcal F(H))n^{c^2})$ for some computable function $f$.
\end{proposition}
\begin{proof}
First, we search for all appearances of homomorphic images of graphs in $\mathcal F(H)$ (there are at most $f(\mathcal F(H))$ such images for some exponential function $f$), which we call \emph{obstructions}. This takes time at most $O(f(\mathcal F(H))n^{c})$, where $c=\max\{|V(F)|, F\in\mathcal F(H)\}$. Then, we need to get rid of each obstruction. For \VDEL{$H$} (resp. \EDEL{$H$}), we need to delete at least one vertex (resp. edge) in each obstruction, thus we can branch on all $c$ (resp. $c^2$) possibilities. For \SWITCH{$H$}, we need to switch at least one of the vertices of the obstruction (but then update the list of obstructions, as we may have created a new one). In all cases, this gives a search tree of height $k$ and degree bounded by a function of $\mathcal F(H)$, which is FPT.
\end{proof}

Some dualities have been obtained for small edge-coloured graphs. The following theorem from~\cite[Section 3]{BBM05} is crucial for our techniques.

\begin{theorem}[{Brewster et al. \cite[Section 3]{BBM05}}]\label{thm:dualities-order2}
Let $H$ be an edge-coloured graph of order at most $2$. Then, $H$ has the polynomial duality property. If $H$ has order $1$, then $H$ has the finite duality property.
\end{theorem}

We next describe the duality sets for some special cases that will be used in our proofs.

\begin{lemma}[{Brewster et al. \cite[Proposition 3.3]{BBM05}}]\label{lemma:duality-H_{r,r}^{2b}}
A $2$-edge-coloured graph has a homomorphism to $H_{r,r}^{2b}$ if and only if it contains no homomorphic image of cycles with an odd number of blue edges. 
\end{lemma}

We present a brief proof of their result. Note that  homomorphic images of paths are walks and that homomorphic images of cycles are closed walks.

\begin{proof}
  Let $G$ be a $2$-edge-coloured graph which admits a homomorphism $\phi$ to $H_{r,r}^{2b}$. Suppose that $G$ contains a homomorphic image of some cycle with an odd number of blue edges, that is to say $G$ contains a closed walk $W$ with an odd number of blue edges. Note that  if $uv$ is a blue edge, then $\phi(u) \neq \phi(v)$ and if $uv$ is a red edge, then $\phi(u) = \phi(v)$. By going around the closed walk, we obtain $\phi(u) \neq \phi(u)$ for any vertex $u$ of $W$, a contradiction.
  
  Let $G$ be a $2$-edge-coloured graph which contains no homomorphic image of cycles with an odd number of blue edges. We identify every connected red components of $G$. The graph that we obtain has red loops but no other red edges, moreover the graph induced by the blue components is bipartite (otherwise there would be a cycle with an odd number of blue edges in $G$). Hence by identifying the vertices of each part in the bipartition, we obtain $H_{r,r}^{2b}$. Hence $G \ecto H_{r,r}^{2b}$.
\end{proof}

\begin{lemma}[{Brewster et al. \cite[Proposition 3.4]{BBM05}}]\label{lemma:duality-H_{r,b}^{2b}}
A $2$-edge-coloured graph has a homomorphism to $H_{r,b}^{2b}$ if and only if it contains no homomorphic image of a red-blue-red 4-vertex path.
\end{lemma}

\begin{proof}
Let $u$ be the vertex of $H_{r,b}^{2b}$ with a red loop, and $v$ the vertex with a blue loop. Given a $2$-edge-coloured graph $G$, map all the vertices incident with a red edge to $u$, and map all others to $v$. This is a homomorphism unless two vertices mapped to $u$ are joined by a blue edge. But in this case, we can find a homomorphic image of a red-blue-red walk in $G$.
Conversely, note that a red-blue-red path has no homomorphism to $H_{r,b}^{2b}$.
\end{proof}

\begin{lemma}[{Brewster et al. \cite[Theorem~3.5]{BBM05}}]\label{lemma:duality-H_{r,-}^{2b}}
A $2$-edge-coloured graph has a homomorphism to $H_{r,-}^{2b}$ if and only if it contains no homomorphic image of a path of the form $RB^{2p-1}R$ (where $R$ is a red edge, $B$ a blue edge and $p\geq 1$ is an integer) or of cycles with an odd number of blue edges.
\end{lemma}

\begin{proof}[Proof (sketch)]
  First note that none of the two obstructions admit a homomorphism to $H_{r,-}^{2b}$. If a $2$-edge-coloured graph $G$ has none of these homomorphic images then by identifying every vertex incident with a red edge of $G$, we obtain a bipartite graph on the blue edges for which one of the two parts contains every vertex incident with a red loop. By mapping this part to the vertex of $H_{r,-}^{2b}$ with the red loop and the other part to the other vertex, we obtain our homomorphism.
\end{proof}

\begin{lemma}[{Brewster et al. \cite[Theorem 3.7]{BBM05}}]\label{lemma:duality-H_{r,r}^{2rb}}
A $2$-edge-coloured graph has a homomorphism to $H_{r,r}^{2rb}$ if and only if it contains no homomorphic image of an all-blue odd cycle.
\end{lemma}

\begin{proof}[Proof (sketch)]
  The idea is to note that the graph induced by the blue edges is bipartite and that the red edges do not create any constraints.
\end{proof}

The proof of the following results are more complicated, hence we refer the reader to~\cite{BBM05} for the details.
In a $2$-edge-coloured graph, a closed walk $v_0v_1\dots v_t$ is \emph{alternating} if for every $i < t-1$, $v_iv_{i+1}$ and $v_{i+1}v_{i+2}$ do not have the same colour. An \emph{odd figure eight} is a closed walk of the form $v_0$, $v_1$,  \dots, $v_{2j}$, $v_0$, $v_{2j+2}$, \dots, $v_{2p-1}$, $v_0$, \textit{i.e.} two odd cycles which share a vertex $v_0$.

\begin{lemma}[{Brewster et al. \cite[Theorem 3.7]{BBM05}}]\label{lemma:duality-H_{r,-}^{2rb}}
A $2$-edge-coloured graph has a homomorphism to $H_{r,-}^{2rb}$ if and only if it contains no homomorphic image of an odd figure eight $v_0$, $v_1$,  \dots, $v_{2j}$, $v_0$, $v_{2j+2}$, \dots, $v_{2p-1}$, $v_0$ for which all edges $v_{2i}v_{2i+1}$ are blue.
\end{lemma}

\begin{lemma}[{Brewster et al. \cite[Theorem 3.7]{BBM05}}]\label{lemma:duality-H_{r,b}^{2rb}}
A $2$-edge-coloured graph has a homomorphism to $H_{r,b}^{2rb}$ if and only if it contains no homomorphic image of \emph{alternating odd figure eight}, that is, an alternating
  closed walk $v_0$, $v_1$,  \dots, $v_{2j}$, $v_0$, $v_{2j+2}$, \dots, $v_{2p-1}$, $v_0$.
\end{lemma}





\section{\Ptime/\NP-complete complexity dichotomies}\label{sec:classical}


In this section, we prove some results about the classical complexity of \VDEL{$H$}, \EDEL{$H$} and \SWITCH{$H$}. 
We first adapt a general method from~\cite{LY80} to show that \VDEL{$H$} is either trivial, or \NP-complete in Section~\ref{sec:vdclassification}.

For \EDEL{$H$} and \SWITCH{$H$}, we cannot use this technique (in fact there exist nontrivial \Ptime cases). Thus, we turn our attention to edge-coloured graphs of order~$2$ (note that for every edge-coloured graph $H$ of order at most~$2$, \textsc{$H$-Colouring} is in \Ptime~\cite{BBM05,BDHQ05}). 
Recall that \SWITCH{$H$} is defined only on $2$-edge-coloured graphs, so our focus is on this case (but for \EDEL{$H$} our results hold for any number of colours).
In Section~\ref{sec:edclassification}, we prove a dichotomy result for graphs of order at most~$2$ for the \EDEL{$H$} problem. The \SWITCH{$H$} problem is treated in Section~\ref{sec:swclassification}, where we also prove a dichotomy result.

The twelve $2$-edge-coloured graphs of order at most~$2$ that are cores (up to symmetries of the colours) are depicted in Figure~\ref{fig:order2}. The two colours are red (dashed edges) and blue (solid edges). We use the terminology of~\cite{BBM05}: for $\alpha\in\{-,r,b,rb\}$, the $2$-edge-coloured graph $H^1_\alpha$ is the graph of order $1$ with no loop, a red loop, a blue loop, and both kinds of loops, respectively. Similarly, for $\alpha\in\{-,r,b,rb\}$ and $\beta,\gamma\in\{-,r,b\}$, the graph $H^{2\alpha}_{\beta,\gamma}$ denotes the graph of order $2$ with vertex set $\{0,1\}$. The string $\alpha$ indicates the presence of an edge between $0$ and $1$: no edge, a red edge, a blue edge and both edges for $-$, $r$, $b$ and $rb$, respectively. Similarly, $\beta$ and $\gamma$ denote the presence of a loop at vertices $0$ and $1$, respectively ($-$ for no loop, $r$ for a red loop, $b$ for a blue loop).


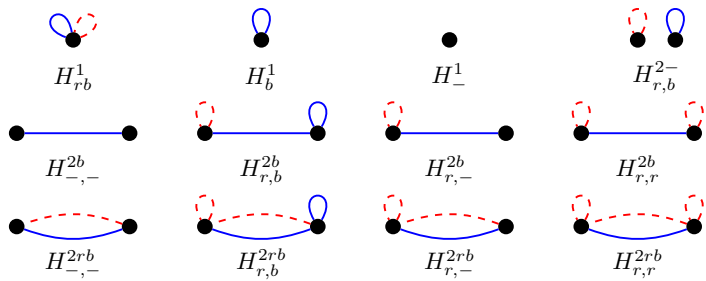
\begin{figure}[!htpb]
\centering
\scalebox{0.9}{\begin{tikzpicture}[every loop/.style={},scale=1.1]
  \begin{scope}[xshift=0.75cm]
    \node[blackvertex] (u) at (0,0) {};
    \path[thick,dashed,red] (u)   edge[out=20,in=80,loop, min distance=5mm] node  {} (u);
    \path[thick,blue] (u)   edge[out=100,in=160,loop, min distance=5mm] node  {} (u);
    
    \path (u)+(0,-0.5) node {$H^{1}_{rb}$};
  \end{scope}

    \begin{scope}[xshift=3.25cm]
    \node[blackvertex] (u) at (0,0) {};
    \path[thick,blue] (u)   edge[out=60,in=120,loop, min distance=5mm] node  {} (u);
    
    \path (u)+(0,-0.5) node {$H^{1}_{b}$};
  \end{scope}
  
  \begin{scope}[xshift=5.75cm]
    \node[blackvertex] (u) at (0,0) {};
    \path (u)+(0,-0.5) node {$H^{1}_-$};
  \end{scope}

  \begin{scope}[xshift=8.25cm]
    \node[blackvertex] (u) at (0,0) {};
    \node[blackvertex] (v) at (0.5,0) {};
    \path[thick,red,dashed] (u)   edge[out=60,in=120,loop, min distance=5mm] node  {} (u);
    \path[thick,blue] (v)   edge[out=60,in=120,loop, min distance=5mm] node  {} (v);
    
    \path (u)+(0.25,-0.5) node {$H^{2-}_{r,b}$};
  \end{scope}

  \begin{scope}[xshift=0cm, yshift=-1.25cm]
    \node[blackvertex] (u) at (0,0) {};
    \node[blackvertex] (v) at (1.5,0) {};
    \draw[thick,blue] (u)--(v);
    \path (u)+(0.75,-0.5) node {$H^{2b}_{-,-}$};
  \end{scope}

  \begin{scope}[xshift=2.5cm,yshift=-1.25cm]
    \node[blackvertex] (u) at (0,0) {};
    \node[blackvertex] (v) at (1.5,0) {};
    \draw[thick,blue] (u)--(v);
    \path[thick,dashed,red] (u)   edge[out=60,in=120,loop, min distance=5mm] node  {} (u);
    \path[thick,blue] (v)   edge[out=60,in=120,loop, min distance=5mm] node  {} (v);
    \path (u)+(0.75,-0.5) node {$H^{2b}_{r,b}$};
  \end{scope}

   \begin{scope}[xshift=5cm,yshift=-1.25cm]
    \node[blackvertex] (u) at (0,0) {};
    \node[blackvertex] (v) at (1.5,0) {};
    \draw[thick,blue] (u)--(v);
    \path[thick,dashed,red] (u)   edge[out=60,in=120,loop, min distance=5mm] node  {} (u);
    \path (u)+(0.75,-0.5) node {$H^{2b}_{r,-}$};
  \end{scope}

   \begin{scope}[xshift=7.5cm,yshift=-1.25cm]
    \node[blackvertex] (u) at (0,0) {};
    \node[blackvertex] (v) at (1.5,0) {};
    \draw[thick,blue] (u)--(v);
    \path[thick,dashed,red] (u)   edge[out=60,in=120,loop, min distance=5mm] node  {} (u);
    \path[thick,dashed,red] (v)   edge[out=60,in=120,loop, min distance=5mm] node  {} (v);
    \path (u)+(0.75,-0.5) node {$H^{2b}_{r,r}$};
   \end{scope}

  \begin{scope}[xshift=0cm, yshift=-2.5cm]
    \node[blackvertex] (u) at (0,0) {};
    \node[blackvertex] (v) at (1.5,0) {};
    \draw[thick,blue] (u) to[bend right=20] (v);
    \draw[thick,dashed,red] (v) to[bend right=20] (u);   
    \path (u)+(0.75,-0.5) node {$H^{2rb}_{-,-}$};
  \end{scope}

  \begin{scope}[xshift=2.5cm,yshift=-2.5cm]
    \node[blackvertex] (u) at (0,0) {};
    \node[blackvertex] (v) at (1.5,0) {};
    \draw[thick,blue] (u) to[bend right=20] (v);
    \draw[thick,dashed,red] (v) to[bend right=20] (u);  
    \path[thick,dashed,red] (u)   edge[out=60,in=120,loop, min distance=5mm] node  {} (u);
    \path[thick,blue] (v)   edge[out=60,in=120,loop, min distance=5mm] node  {} (v);
    \path (u)+(0.75,-0.5) node {$H^{2rb}_{r,b}$};
  \end{scope}

   \begin{scope}[xshift=5cm,yshift=-2.5cm]
    \node[blackvertex] (u) at (0,0) {};
    \node[blackvertex] (v) at (1.5,0) {};
    \draw[thick,blue] (u) to[bend right=20] (v);
    \draw[thick,dashed,red] (v) to[bend right=20] (u);  
    \path[thick,dashed,red] (u)   edge[out=60,in=120,loop, min distance=5mm] node  {} (u);
    \path (u)+(0.75,-0.5) node {$H^{2rb}_{r,-}$};
  \end{scope}

   \begin{scope}[xshift=7.5cm,yshift=-2.5cm]
    \node[blackvertex] (u) at (0,0) {};
    \node[blackvertex] (v) at (1.5,0) {};
    \draw[thick,blue] (u) to[bend right=20] (v);
    \draw[thick,dashed,red] (v) to[bend right=20] (u); 
    \path[thick,dashed,red] (u)   edge[out=60,in=120,loop, min distance=5mm] node  {} (u);
    \path[thick,dashed,red] (v)   edge[out=60,in=120,loop, min distance=5mm] node  {} (v);
    \path (u)+(0.75,-0.5) node {$H^{2rb}_{r,r}$};
  \end{scope}

  
  

  



  \end{tikzpicture}}
\caption{The twelve $2$-edge-coloured cores of order at most~$2$ considered in this paper.}\label{fig:order2}  
\end{figure}

\subsection{Dichotomy for \VDEL{$H$} for all $H$}
\label{sec:vdclassification}

Graph modification problems for operations vertex-deletion and edge-deletion have been studied extensively. For a graph property $\mathcal P$, we denote by \textsc{Vertex Deletion-${\mathcal{P}}$} the graph modification problem for property $\mathcal P$ and operation vertex-deletion. A property is {\em hereditary} if $\mathcal P(G)$ implies $\mathcal P(H)$ for all induced subgraphs $H$ of $G$.
Lewis and Yannakakis~\cite{LY80} defined a non-trivial property $\mathcal P$ on graphs as a property true for infinitely many graphs and false for infinitely many graphs. These definitions can be extended to edge-coloured graphs. They showed the following general result.

\begin{theorem}[Lewis and Yannakakis \cite{LY80}]
	\label{th:VDisNP}
	The \textsc{Vertex Deletion-${\mathcal{P}}$} problem for non-trivial graph-properties $\mathcal P$ that are
	hereditary is \NP-hard. 
\end{theorem}

By modifying the proof of Theorem~\ref{th:VDisNP}, we can prove the two following results.

\begin{theorem}
	\label{cor:VDEL-NPC}
	Let $\mathcal P$ be a non-trivial property of $t$-edge-coloured graphs that is hereditary and true for all empty graphs. Then, \textsc{Vertex Deletion-${\mathcal{P}}$} is \NP-hard.
\end{theorem}

The proof of this theorem follows the proof of Theorem~\ref{th:VDisNP} from~\cite{LY80}. The only difference is that we work with $t$-edge-coloured graphs instead of undirected graphs.

\begin{proof}
	Let $G$ be a $t$-edge-coloured graph. We denote by $CC(G)$ the set of connected components of $G$. These components are also $t$-edge-coloured graphs. For $x$ and $v$ two vertices of $G$, let $R_v(x)$ be the set of vertices connected to $x$ in $G - v$. For any vertex $v \in V(G)$, let $CC_v(G)$ be the set of connected subgraphs of $G$ induced by the sets of vertices  of the form $R_v(x) \cup \set{x}$ for $x \in V(G-v)$. In other words, $CC_v(G)$ is the set of connected components of $G-v$ where we added the vertex $v$.
	In particular, if $v$ is not a cut-vertex, then $CC_v(G) = \set{G}$.
	
	For a connected $t$-edge-coloured graph $G$ and $v \in G$, let $\alpha_v(G) = (n_1,n_2, \dots n_{t})$ such that $n_1 \geq n_2 \geq \dots \geq n_{t}$ and the multi-sets $\set{n_1,\dots,n_t}$ and  $\set{\abs{V(C)}\ :\ C \in CC_v(G)}$ are equal. In other words, $\alpha_v(G)$ is the ordered sequence of the orders of the $t$-edge-coloured graphs in $CC_v(G)$. Let $\alpha(G)$ be the smallest sequence (for the lexicographic order) $\alpha_v(G)$ over all possible vertices $v \in V(G)$.
	
	For a $t$-edge-coloured graph $G$, let $\beta(G) = (\alpha(G_1),\alpha(G_1), \dots \alpha(G_t))$ such that $\alpha(G_1) \geq_L \alpha(G_2) \geq_L \dots \geq_L \alpha(G_t)$ (where $\geq_L$ is the lexicographical order) and $ CC(G) = \set{G_1,\dots,G_t}$. In other words, $\beta(G)$ is the ordered sequence of $\alpha$-sequences of the connected components of~$G$.
	
	Recall that $\mathcal{P}$ is non-trivial. In particular, $\mathcal P$ has counter-examples. For an integer $p$ and a $t$-edge-coloured graph $G$, we denote by $pG$, the $t$-edge-coloured graph composed of $p$ disjoint copies of $G$. Let $J$ be a $t$-edge-coloured graph such there exists some $k \geq 1$ for which $\mathcal P(kJ)$ is false, and which has the minimum $\beta$-sequence among the $t$-edge-coloured graphs verifying this property. Let $k \geq 1$ such that $\mathcal P(kJ)$ is false and $\mathcal P((k-1)J)$ is true.
	Suppose that $\beta(J) = (\alpha(J_1),\dots, \alpha(J_t))$ where $CC(J) = \set{J_1,\dots,J_t}$. Let $x$ be a vertex of $J_1$ for which $\alpha(J_1) = \alpha_x(J_1)$ and let $J^+$ be the connected $t$-edge-coloured graph of $CC_x(J_1)$ with the greatest number of vertices. Since all empty graphs verify $\mathcal P$, $J$ contains at least one edge. This implies that $J_1$ and $J^+$ contain at least one edge. In particular, $J^+$ contains at least two vertices. Let $y$ be a vertex of $J^+$ which is different from $x$.
	Let $J_1'$ (resp. $J'$) be the $t$-edge-coloured graph obtained from $J_1$ (resp. $J$) by removing the vertices of $V(J^+) \setminus \set{x}$. See Figure~\ref{fig:homfpt:J} for an example.
	
	\begin{figure}
		\centering
		\subcaptionbox{Some possible $4$-edge-coloured graph $J$.}[.32\linewidth]{
			\begin{tikzpicture}
				
				\node[blackvertex] (a) at (2,0) {};
				\node[blackvertex] (b) at (2,1) {};
				\draw[blue] (a) -- (b);
				\node (x) at (2.35,0.5) {$J_2$};
				
				\node[blackvertex] (c) at (0,0.5) {};
				\node[blackvertex] (d) at (0,-0.5) {};
				\node[blackvertex,label=left:$y$] (e) at (0-0.5,1.5) {};
				\node[blackvertex] (f) at (0.5,1.5) {};
				\node (x) at (0.25,0.5) {$x$};
				\node (x) at (-0.45,0.5) {$J_1$};
				\draw (c) edge[blue, bend right] (d) edge[-, black, bend left, densely dotted] (d);
				\draw (c) edge[blue] (e) edge[-,  red, dashed] (f);
				\draw (e) edge[teal, dashdotted] (f);
			\end{tikzpicture}
		}\qquad
		\subcaptionbox{The $4$-edge-coloured graph $J^+$.}[.32\linewidth]{
			\begin{tikzpicture}
				
				\node[blackvertex,label=left:$y$] (e) at (0-0.5,1.5) {};
				\node[blackvertex] (f) at (0.5,1.5) {};
				\node[blackvertex] (c) at (0,0.5) {};
				\node (x) at (0.25,0.5) {$x$};
				\draw (c) edge[blue] (e) edge[-, red,dashed] (f);
				\draw (e) edge[teal, dashdotted] (f);
			\end{tikzpicture}
		}\qquad
		\subcaptionbox{The $4$-edge-coloured graph $J'$.}[.32\linewidth]{
			\begin{tikzpicture}
				
				\node[blackvertex] (a) at (2,0-.5) {};
				\node[blackvertex] (b) at (2,.5) {};
				\draw[blue] (a) -- (b);
				\node (x) at (2.35,0) {$J_2$};
				
				\node[blackvertex] (c) at (0,0.5) {};
				\node[blackvertex] (d) at (0,-0.5) {};
				\node (x) at (0.25,0.5) {$x$};
				\node (x) at (-0.45,0) {$J'_1$};
				\draw (c) edge[blue, bend right] (d) edge[-, black, bend left, densely dotted] (d);
			\end{tikzpicture}
		}
		\caption{An example of a $4$-edge-coloured graph $J$ and its induced subgraphs $J_1$, $J_2$, $J^+$, $J'_1$ and $J'$.}
		\label{fig:homfpt:J}
	\end{figure}
	
	Each induced subgraph of $J$ that we defined will be useful to show that \textsc{Vertex Deletion-${\mathcal{P}}$} is \NP-hard. 
	We reduce \textsc{Vertex Cover} to \textsc{Vertex Deletion-${\mathcal{P}}$}. (Note that it may be complicated to find the $t$-edge-coloured graph $J$, but this is a constant graph that depends only on $\mathcal{P}$ and this fact only makes the reduction non-constructive.) Let $(G,\ell)$ be an instance of \textsc{Vertex Cover} where $G$ is an undirected graph of order $p$ and $\ell$ is an integer.
	
	We construct the $t$-edge-coloured graph $H$ from $G$ as follows. For each vertex $v \in V(G)$, we add a copy $J'_v$ of $J'$ to $H$. 
	For each edge $uv \in E(G)$, we add a copy $J^+_{uv}$ of $J^+$ to $H$. We identify the copy $x_u$ (resp. $x_v$) of $x$ in $J'_u$ (resp. $J'_v$)  with the copy $x_{uv}$ (resp. $y_{uv}$) of $x$ (resp. $y$) in $J^+_{uv}$.
	This concludes the construction of $H$. See Figure~\ref{fig:homfpt:yannakakisEdge} for an example. We construct the $t$-edge-coloured graph $H'$ by taking $pk$ disjoint copies of $H$.
	
	\begin{figure}
		\centering
		\begin{tikzpicture}
			
			\node[blackvertex] (a) at (4,0) {};
			\node[blackvertex] (b) at (4,1) {};
			\draw[blue] (a) -- (b);
			\node (x) at (4.35,0.5) {$J_2$};
			
			\node[blackvertex] (a) at (-2,0) {};
			\node[blackvertex] (b) at (-2,1) {};
			\draw[blue] (a) -- (b);
			\node (x) at (-2.35,0.5) {$J_2$};
			
			\node[blackvertex] (c1) at (0,00) {};
			\node[blackvertex] (d1) at (0,1) {};
			\draw (c1) edge[blue, bend right] (d1) edge[-, black, densely dotted, bend left] (d1);
			\node (x) at (-0.45,0.5) {$J'_1$};
			\node (x) at (-0.5,0) {$x_u$};
			
			\node[blackvertex] (c2) at (2,00) {};
			\node[blackvertex] (d2) at (2,1) {};
			\draw (c2) edge[blue, bend right] (d2) edge[-, black, densely dotted, bend left] (d2);
			\node (x) at (2.45,0.5) {$J'_1$};
			\node (x) at (2.5,0) {$x_v$};
			
			\node[blackvertex] (f) at (1,-0.5) {};
			\node (x) at (1,0.5) {$J^*_{uv}$};
			\draw (c1) edge[blue] (c2) edge[-, red, dashed] (f);
			\draw (c2) edge[teal, dashdotted] (f);
			
			\draw (1.5,1.25) -- (1.5,-0.25) -- (5,-0.25) -- (5,1.25) -- cycle;
			\node (x) at (5.5,0.5) {$J'_v$};
			\draw (0.5,1.25) -- (0.5,-0.25) -- (-3,-0.25) -- (-3,1.25) -- cycle;
			\node (x) at (-3.5,0.5) {$J'_u$};
			
		\end{tikzpicture}
		\caption{An example of the graph $H$ when $J$ is the $3$-edge-coloured graph of Figure~\ref{fig:homfpt:J} and $G$ is just an edge $uv$. Here, we chose to identify $x_u$ with $x_{uv}$ and $x_v$ with $y_{uv}$. Note that if no vertex is removed from $H$, then $H$ contains $J$ as an induced subgraph.}
		\label{fig:homfpt:yannakakisEdge}
	\end{figure}
	
	We claim that $(G,\ell)$ is a positive instance of \textsc{Vertex Cover} if and only if $(H',pk\ell)$ is a positive instance of \textsc{Vertex Deletion-${\mathcal{P}}$}.
	
	Suppose that there is a subset $S$ of vertices of $G$ of size at most $\ell$ that is a vertex cover of $G$.
	We construct  $S' \subseteq V(H')$ as follows. For every copy of $H$ in $H'$ and every vertex $u \in S$, we add the copy of the vertex $x_u$ of $J'_u$ to $S'$. Note that $\abs{S'} \leq pk\ell$. We claim that $H' - S'$ verifies $\mathcal P$. 
	Let $\mathcal J$ be the set of $t$-edge-coloured graphs that can be constructed as follows. Take a copy of $J'_1$ and at most $\Delta(G)$ copies of $J^+$. For each copy of $J^+$, delete one of $x$ or $y$ and identify the other vertex with the copy $x'$ of $x$ in the copy of $J'_1$. The set $\mathcal J$ contains at most $3^{\Delta(G)}$ $t$-edge-coloured graphs, $\mathcal J$ contains all possible maximal connected induced subgraphs of $H$ connected to a vertex $x_u$ when every $x_v$ for $v \in N(u)$ has been removed in $H$.
	
	A connected component $C$ of $H' - S'$ can be of four types.
	\begin{enumerate}
		\item \emph{The connected component $C$ belongs to $\set{J_2,\dots, J_t}$.} 
		\item \emph{The connected component $C$ belongs to $\mathcal J$.} 
		\item \emph{The connected component $C$ is isomorphic to a connected induced subgraph of $J'_1$ where the vertex $x$ has been removed.} 
		\item \emph{The connected component $C$ is isomorphic to a connected induced subgraph of $J^+$ where the vertices $x$ and $y$ have been removed.}
	\end{enumerate}
	
	Let $J^*$ be the $t$-edge-coloured graph composed of disjoint copies of the vertices of $\mathcal J$ and disjoint copies of $J_2,\ldots,J_{t-1}$ and $J_t$. Note that every connected component of $H' - S'$ is an induced subgraph of $J^*$. Let $C \in \mathcal J$, note that $\alpha(C) \leq \alpha_{x'}(C)$ where $x'$ is the copy of $x$ in $J'_1$. Note that $CC_{x'}(C) = CC_{x}(J') \cup X$ where $X$ is the set corresponding to the copies of $J^+$ in $C$ with one of $x$ or $y$ removed. The connected multigraphs of $X$ have order $\abs{V(J^+)} -1$, hence $\alpha_{x'}(C) <_L \alpha_x(J_1)= \alpha(J_1)$.
	Note that $\beta(J^*) <_L \beta(J)$ since for every $C \in \mathcal J$, $\alpha(C) <_L \alpha(J_1)$.
	
	By minimality of $J$, any number of disjoint copies of $J^*$ must verify $\mathcal P$, hence $H' - S'$ verifies $\mathcal P$ and $(H',pk\ell)$ is a positive instance of \textsc{Vertex Deletion-${\mathcal{P}}$}.
	
	Suppose that there is a subset $S'$ of vertices of $H'$ of size at most $pk\ell$ such that $\mathcal P(H'-S')$ holds.
	Note that $H'-S'$ can contain at most $k-1$ copies of the $t$-edge-coloured graph $J$ by definition of $J$. In particular $H'$ has at least $pk - (k-1)$ copies of $H$ for which after removing the vertices of $S'$, the $t$-edge-coloured graph does not contain a copy of $J$. 
	
	Suppose that for one of the copies $H_0$ of $H$, $\abs{V(H_0) \cap S'} \leq \ell$. In this case, we construct $S \subseteq V(G)$ as follows. If $S' \cap V(J'_u) \neq \varnothing$, then add $u$ to $S$. If $S' \cap (V(J^+_{uv}) \setminus \set{x,y}) \neq \varnothing$, then add arbitrarily one of $u$ or $v$ to $S$. Note that $\abs{S} \leq \ell$. Suppose that there is an edge $uv \in E(G)$, such that $u,v \notin S$. Our copy of $H$ contains $J'_u$, $J'_v$ and $J^+_{uv}$ and these  $t$-edge-coloured graphs do not contain vertices from $S$. The vertex $x_{uv}$ has been identified with one of $x_u$ or $x_v$, say $x_u$. The  $t$-edge-coloured graph composed of $J'_u$ and $J^+_{uv}$ with $x_u$ and $x_{uv}$ identified is exactly the $t$-edge-coloured graph $J$. Hence if $H - S'$ does not contain $J$, the set $S$ is a vertex cover of $G$ of size at most $\ell$. 
l	
	Suppose, by contradiction, that for every copy of $H$ either $H - S'$ contains $J$ or verifies $\abs{V(H_0) \cap S'} \geq \ell +1$. In this case, $S'$ has at least $(pk - (k-1)) (\ell +1)$ vertices. Moreover, as $\ell < p$ (otherwise the instance of \textsc{Vertex Cover} is trivial),  $(pk - (k-1)) (\ell +1) \geq pk\ell + \ell +1 + k(p - (\ell -1)) > pk\ell$, a contradiction.
	
	Hence $G$ has a vertex cover of size at most $\ell$.
\end{proof}

For a $t$-edge-coloured graph, the only case where the property of mapping to $H$ is trivial (in this case, always true) is when $H$ has a vertex with all $t$ kinds of loops attached (in which case the core of $H$ is the subgraph induced by that vertex). Thus we obtain the following dichotomy.

\begin{corollary}
\label{cor:VDEL-classification}
Let $H$ be a $t$-edge-coloured graph. \VDEL{$H$} is in \Ptime if $H$ contains a vertex having a loop of each edge-colour, and \NP-complete otherwise.
\end{corollary}
\begin{proof}
For every edge-coloured graph $H$, \VDEL{$H$} is in \NP.
\textsc{$H$-colouring} is a hereditary property and is verified by all independent sets, thus if it has infinitely many NO-instances (on loopless $t$-edge-coloured graphs); it is nontrivial, and thus \NP-hard by Theorem~\ref{cor:VDEL-NPC}. Let us see when this is the case.

We can observe that the problem is actually trivial if $H$ contains a vertex with all $t$-coloured loops, indeed every $t$-edge-coloured graph can be mapped to this vertex (in this case, we return YES). Moreover, if not, then the complete graph $K_{|H|+1}$ with all $t$-coloured edges between each pair of vertices does not map to $H$. Indeed by the pigeonhole principle, two vertices $u$ and $v$ of our input $t$-edge-coloured graph must have the same image vertex $w$ in $H$. As there is an edge coloured $i$ between $u$ and $v$, there must be a loop coloured $i$ on $w$. Thus $w$ should have all $t$-coloured loops, a contradiction. Thus, in all such cases, the property is nontrivial on loopless $t$-edge-coloured graphs and hence the problem is \NP-complete.
\end{proof}

\subsection{Dichotomy for \EDEL{$H$} when $H$ has order~$2$}
\label{sec:edclassification}

No analogue of Theorem~\ref{th:VDisNP} for operation edge-deletion exists nor is expected to exist~\cite{Y81}. We thus restrict our attention to the case of edge-coloured graphs $H$ of order at most~$2$. For this case we classify the complexity of \EDEL{$H$}. Since multiple edges of the same colour are irrelevant, if $H$ has order $2$, for each edge-colour there are three possible edges.

\begin{theorem}\label{thm:EDEL-order2-poly/NPc}
Let $H$ be an edge-coloured core of order at most~$2$. If each colour class of the edges of $H$ contains only loops or contains all three possible edges, then \EDEL{$H$} is in \Ptime; otherwise it is \NP-complete.
\end{theorem}

We separate the proof of this theorem into several lemmas.

\begin{lemma}
Let $H$ be an edge-coloured core of order at most~$2$. If each colour class of the edges of $H$ contains only loops or contains all three possible edges, then \EDEL{$H$} is in \Ptime.
\end{lemma}

\begin{proof}
First note that if colour $i$ has all three possible edges in $H$, we can simply ignore this colour by removing it from $H$ and $G$ without decreasing the parameter, as it does not provide any constraint on the homomorphisms.
  
  We can therefore suppose that $H$ contains only loops. If two colours induce the same subgraph of $H$, then we can identify these two colours in both $G$ and $H$ as they give the same constraints.
  
  If $G$ has colours that $H$ does not have, then remove each edge with this colour and decrease the parameter for each removed edge. If it goes below zero then we reject.
  
  We can now assume that $H$ has only loops and $G$ has the same colours as $H$. We are left with only a few cases, as $H$ is a core (there is no vertex whose set of loops is included in the set of loops of the other).
  \begin{itemize}
  \item $H$ has a single loop. Then, $G \ecto H$ as $G$ has the same colours as $H$.
  \item $H$ contains two non-incident loops with different colours and two non-incident loops of a third colour.  Up to symmetry, suppose that $H$ has one blue loop and one green loop on the first vertex and has one red loop and one green loop on the second vertex. We will reduce to the problem where we have removed the green loops. Let $p$ be the number of green edges of $G$. We construct $G'$ from $G$ by replacing each green edge by a blue edge and a red edge (we can end up with multiple blue or red edges that way). We claim that \EDEL{$H$} with parameter $k$ and input $G$ is true if and only if \EDEL{$H^{2-}_{r,b}$} with parameter $k +p$ on input $G'$ is true.

  If the first problem has a solution $S$, then remove the corresponding edges from $G'$ (if the corresponding edge of $G$ is green remove the two new edges in $G'$). Each vertex of $G-S$ is set to one component, in particular each green edge is set to a vertex with a blue edge or a red edge. If a green edge $uv$ of $G$ is sent to the first vertex, we remove the edge of $G'$ corresponding to $uv$ which is red. We can check that after removing those edges, $G'$ admits a homomorphism to  $H^{2-}_{r,b}$. We removed at most $k$ edges in the first step plus the number of green edges in $S$ and removed one edge for each green edge left in the second step. Thus, we removed less than $k +p$ edges in $G'$.
  
  If the second problem has a solution $S$, then remove from $G$ all blue and red edges of $S$. Remove the green edges of $G$ only if both were removed in $G'$. Note that $S$ contains at least one edge in $G'$ for each green edge of $G$. Thus we removed less than $k$ edges in $G$. Moreover, $G \ecto H$ by taking the same homomorphism as in $G'$. Indeed, the blue and red edges are sent to one of the two loops while each green connected component is sent to one vertex.

  Using this method we can reduce the problem to \EDEL{$H^{2-}_{r,b}$}, which is our last case.
  
  \item $H$ contains two non-incident loops with different colours; then $H = H^{2-}_{r,b}$. Indeed if there were any other kind of loop, then we would be in the previous case or we could identify two colours. Note that a $2$-edge-coloured graph maps to $H^{2-}_{r,b}$ if and only if it has no red edge incident to a blue edge. Thus,
    solving \EDEL{$H^{2-}_{r,b}$} amounts to splitting $G$ into disconnecting red and blue connected components. This can be done by constructing the following bipartite graph: put a vertex for each edge of $G$; two vertices are adjacent if the corresponding edges in $G$ are adjacent and of different colours. Solving \EDEL{$H^{2-}_{r,b}$} is the same as solving \textsc{Vertex Cover} on this bipartite graph, which is in \Ptime.
  \end{itemize}
  There is no other case as otherwise the set of loops of one vertex would be included in the set of loops of the other. 
\end{proof}

The \NP-completeness proofs are by reductions from \textsc{Vertex Cover}, based on vertex- and edge-gadgets constructed using obstructions to the corresponding homomorphisms from~\cite{BBM05} presented in Section~\ref{sec:duality}.

We start with proving the \NP-hardness of two special cases, and then we will show that we can always reduce the problem from these two cases.

\begin{lemma}\label{lemma11}
The problem \EDEL{$H^{2b}_{r,b}$} is \NP-hard.
\end{lemma}

\begin{proof}
 We reduce from \textsc{Vertex Cover}. Given an input graph
  $G$ of \textsc{Vertex Cover}, we construct a $2$-edge-coloured graph
  $G'$ from $G$ as follows. Take $G$ and colour all edges blue, then
  add a pending red edge $vv'$ to each vertex $v$ of $G$ (see
  Figure~\ref{fig:reduc-H^{2b}_{r,b}}).
  \begin{figure}[!htpb]
    \centering
    \scalebox{0.8}{\begin{tikzpicture}[every loop/.style={},scale=1.2]
      \node[blackvertex,label={270:$u$}] (u) at (2,0) {};
      \node[blackvertex,label={270:$v$}] (v) at (3,0) {};
      \node[blackvertex,label={270:$w$}] (w) at (4,0) {};
      \node[blackvertex,label={0:$u'$}] (u') at (2,1) {};
      \node[blackvertex,label={0:$v'$}] (v') at (3,1) {};
      \node[blackvertex,label={0:$w'$}] (w') at (4,1) {};

      \draw[thick,dashed,red] (u') -- (u);
      \draw[thick,dashed,red] (v') -- (v);
      \draw[thick,dashed,red] (w') -- (w);
      
      \draw[thick,blue] (u) -- (v);
      \draw[thick,blue] (u) to[bend right=40] (w);
      
      \path (0,0) node {$G$};
      \path (1,0) node {\ldots};
      
      \draw[line width=1.1pt] (2,0) ellipse (2.5cm and 0.75cm);  
    \end{tikzpicture}}
    \caption{Reduction from \textsc{Vertex Cover} to  \EDEL{$H^{2b}_{r,b}$}.}\label{fig:reduc-H^{2b}_{r,b}}
  \end{figure}
  By Lemma~\ref{lemma:duality-H_{r,b}^{2b}}, a $2$-edge-coloured graph maps to $H^{2b}_{r,b}$ if and
  only if it does not contain a homomorphic image of a red-blue-red
  $3$-edge-path~\cite{BBM05}.

  Assume that $G$ has a vertex cover $C$ of size at most $k$. When
  removing these vertices in $G'$, the resulting graph is a collection
  of independent red edges and thus maps to $H^{2b}_{r,b}$.
  
  Conversely, assume that we have a set $S$ of $k$ 
   edges of $G'$ such that $(G'-S)\ecto H^{2b}_{r,b}$. In
  particular, for every blue edge $uv$ of $G$, we must have one of
   $uu', uv$ or $vv'$ in
  $S$. Thus we can obtain a vertex cover of $G$ of size $k$ from $S$:
  for a vertex $v$, if  $vv'$ belongs to $S$, we
  add $v$ to that vertex cover. If $uv\in S$, we add
  randomly $u$ or $v$ to the vertex cover.

  We thus have a polynomial-time reduction from \textsc{Vertex Cover}
  to \EDEL{$H^{2b}_{r,b}$}. Therefore this
  problem is \NP-hard.
\end{proof}

\begin{lemma}\label{lemma12}
The problem \EDEL{$H^{2rb}_{r,b}$} is \NP-hard.
\end{lemma}

\begin{proof}
We again reduce from \textsc{Vertex Cover}. For an input graph $G$ of
  \textsc{Vertex Cover}, we construct a $2$-edge-coloured graph $G'$
  from $G$ as follows. We start with a red copy of $G$, then we add a
  pending blue edge $vv'$ for each $v\in G$. Finally, for each edge
  $uv\in G$, we create three new vertices $x_{uv},y_{uv},z_{uv}$ such
  that $u'x_{uv},v'x_{uv},y_{uv}z_{uv}$ are red and
  $x_{uv}y_{uv},x_{uv}z_{uv}$ are blue (see
  Figure~\ref{fig:reduc-VD-H^{2rb}_{r,b}}).

  We then recall Lemma~\ref{lemma:duality-H_{r,b}^{2rb}} proved in~\cite{BBM05}, stating that a
  $2$-edge-coloured graph maps to $H^{2rb}_{r,b}$ if and only if it
  does not contain an \emph{alternating odd figure eight}, that is, an alternating
  closed walk $v_0$, $v_1$,  \dots, $v_{2j}$, $v_0$, $v_{2j+2}$, \dots, $v_{2p-1}$, $v_0$. Note that our construction creates such
  a pattern for each edge of $G$. 

  \begin{figure}[!htpb]
    \centering
        \scalebox{0.8}{\begin{tikzpicture}[every loop/.style={},scale=1.2]

      \node[blackvertex,label={270:$u$}] (u) at (-1,0) {};
      \node[blackvertex,label={270:$v$}] (v) at (0,0) {};
      \node[blackvertex,label={270:$w$}] (w) at (1,-0.5) {};
      \node[blackvertex,label={270:$t$}] (t) at (1,0.5) {};
      \path (0.5,-1.5) node {$G$};

      \draw[thick] (u) -- (v) -- (w);
      \draw[thick] (v) -- (t);
      
      \draw[decorate, decoration={snake},->] (2,0) -- (3,0);

      \begin{scope}
        
        \node[blackvertex,label={270:$u$}] (u) at (4,0) {};
        \node[blackvertex,label={180:$u'$}] (u') at (4,1) {};
        \node[blackvertex,label={180:$x_{uv}$}] (xuv) at (4.5,1.5) {};
        \node[blackvertex,label={90:$y_{uv}$}] (yuv) at (4.2,2) {};
        \node[blackvertex,label={90:$z_{uv}$}] (zuv) at (4.8,2) {};
        \node[blackvertex,label={270:$v$}] (v) at (5,0) {};
        \node[blackvertex,label={180:$v'$}] (v') at (5,1) {};
        \node[blackvertex,label={0:$x_{vt}$}] (xvt) at (5.5,2) {};
        \node[blackvertex,label={90:$y_{vt}$}] (yvt) at (5.2,2.5) {};
        \node[blackvertex,label={90:$z_{vt}$}] (zvt) at (5.8,2.5) {};
        \node[blackvertex,label={270:$w$}] (w) at (6,-0.5) {};
        \node[blackvertex,label={0:$w'$}] (w') at (7,0) {};
        \node[blackvertex,label={270:$t$}] (t) at (6,0.5) {};
        \node[blackvertex,label={0:$t'$}] (t') at (6,1.5) {};
        \node[blackvertex,label={0:$x_{vw}$}] (xvw) at (7,1) {};
        \node[blackvertex,label={90:$y_{vw}$}] (yvw) at (7,1.5) {};
        \node[blackvertex,label={90:$z_{vw}$}] (zvw) at (7.5,1.5) {};

        \path (5,-1.5) node {$G'$};

        \draw[thick,blue] (u) -- (u');
        \draw[thick,blue] (v) -- (v');
        \draw[thick,blue] (w) -- (w');
        \draw[thick,blue] (t) -- (t');
        \draw[thick,red,dashed] (u) -- (v) -- (w);
        \draw[thick,red,dashed] (v) -- (t);
        \draw[thick,red,dashed] (v') -- (xuv) -- (u');
        \draw[thick,red,dashed] (yuv) -- (zuv);
        \draw[thick,blue] (yuv) -- (xuv) -- (zuv);
        \draw[thick,red,dashed] (v') -- (xvt) -- (t');
        \draw[thick,red,dashed] (yvt) -- (zvt);
        \draw[thick,blue] (yvt) -- (xvt) -- (zvt);
        \draw[thick,red,dashed] (v') -- (xvw) -- (w');
        \draw[thick,red,dashed] (yvw) -- (zvw);
        \draw[thick,blue] (yvw) -- (xvw) -- (zvw);
      \end{scope}

    \end{tikzpicture}}
    \caption{Reduction from \textsc{Vertex Cover} to \EDEL{$H^{2rb}_{r,b}$}}
    \label{fig:reduc-VD-H^{2rb}_{r,b}}
  \end{figure}

  Assume that $G$ has a vertex cover $C$ of size at most $k$. Then for
  each $v\in C$, we delete $vv'$ from $G'$. We prove that
  the resulting graph $G''$ contains no alternating odd figure eight. First
  observe that in the graph obtained from $G'$ by removing all edges
  from $G$, all the alternating walks have length at most $7$, hence
  it contains no odd figure eight. Thus, if $G''$ contains an alternating  odd
  figure eight, then it uses an edge $uv$ from $G$. Since $C$ is a
  vertex cover, either $uu'$ or $vv'$ is not
  present in $G''$. Then, either $u$ or $v$ has no incident blue
  edge. This implies that $G''$ has no alternating odd figure eight, and hence
  maps to $H^{2rb}_{r,b}$.

  Conversely, assume that we can remove a set $S$ of $k$ edges from $G'$ so that $G'\setminus S\ecto H^{2rb}_{r,b}$. We
  construct a set $C\subset V(G)$ as follows: if
  $vv'\in S$, then we add $v\in C$. If
  $uv,u'x_{uv},v'x_{uv},x_{uv}y_{uv},x_{uv}z_{uv}$ or $y_{uv}z_{uv}$
  lie in $S$, then we add randomly $u$ or $v$ to $C$. 
  Note that, in
  each case, $|C|\leqslant k$. Moreover, we claim that $C$ is a vertex
  cover of $G$. Assume not, and consider an edge $uv$ in $G$ such that
  $u,v\notin C$. By construction, this means that none of the edges
  $uv,uu',vv',u'x_{uv},v'x_{uv},x_{uv}y_{uv},x_{uv}z_{uv},y_{uv}z_{uv}$ lies in $S$. These
  vertices form an alternating odd figure eight, contradicting that
  $G'\setminus S\ecto H^{2rb}_{r,b}$.

  Therefore, \EDEL{$H^{2rb}_{r,b}$} is 
  \NP-hard.
\end{proof}

\begin{lemma}
 For $H$ an edge-coloured core of order at most $2$, if there exists a colour of $H$ which contains a non-loop and does not contain all three possible edges, then \EDEL{$H$} is \NP-complete. 
\end{lemma}

\begin{proof}
 Take such a graph $H$. If one colour, say blue, contains only one
  edge from the first vertex to the second, then for graphs $G$ which
  are all blue, the problem is equivalent to \textsc{Edge
    Bipartization}, which is \NP-complete.

  Now, if $H$ contains no such edge, then by assumption it must contain
  a colour, say blue, with a loop and an edge from the first vertex to
  the second (and no other edge of this colour). Let $u$ be the vertex with the loop and $v$ be the other vertex. Since $H$ is a core, $H$ does not map to its subgraph induced by $u$. If for every colour of $H$ there was a loop of this colour on $u$, then $H$ would not be a core. Hence there exists a colour, say red, such that there is a red edge in $H$ and $u$ has no loop coloured red. Hence, the graph obtained by removing all edges which are neither blue nor red, is either
  $H^{2b}_{r,b}$ or $H^{2rb}_{r,b}$ up to symmetry.
  %
  Thus, by the previous two Lemmas~\ref{lemma11} and~\ref{lemma12}, the
  problem is \NP-complete using the same reductions (the edges of $H$ that are neither blue nor red can be ignored).
\end{proof}

\subsection{Dichotomy for \SWITCH{$H$} when $H$ has order~$2$}
\label{sec:swclassification}

We now turn our attention to the switching operation.

\begin{theorem}\label{thm:SWITCH-order2-poly/NPc}
Let $H$ be a $2$-edge-coloured graph from Figure~\ref{fig:order2}. If $H$ is one of $H^{2b}_{r,b}$, $H^{2b}_{r,-}$, $H^{2rb}_{r,b}$, $H^{2rb}_{r,-}$ or $H^{2rb}_{r,r}$, then \SWITCH{$H$} is \NP-complete. Otherwise, it is in \Ptime.
\end{theorem}
\begin{proof}  We begin with the \Ptime cases. 
  \begin{itemize}    
  \item Every $2$-edge-coloured graph maps to $H^1_{rb}$, thus \SWITCH{$H^1_{rb}$} is trivially in \Ptime.
    \item No graph with an edge can be mapped to $H^1_{-}$ (regardless of switchings).
\item  For $H^1_b$, we need to test if the graph can be switched to an all-blue graph in less than $k$ switchings. There are only two sets of switchings that achieve this signature (one is the
    complement of the other). It is in \Ptime to test if the
    graph can be switched to an all-blue graph (see~\cite[Proposition~2.1]{SignedDicho}). Doing that also gives us one of the
    two switching sets; we then need to check if its size is at most $k$
    or at least $|V(G)|-k$. So, \SWITCH{$H^1_b$} is in \Ptime.
\item For $H^{2-}_{r,b}$, we just apply the algorithm for
$H^1_b$ and $H^1_r$ to each connected component, one of the two
 must accept for each of them.
\item For $H^{2rb}_{-,-}$, a graph $G$ is a YES-instance if and only if $G$ (without considering edge-colours) is bipartite, which is polynomially testable.
\item For $H^{2b}_{-,-}$ a graph $G$ is a YES-instance if and only if it is bipartite and maps to $H^1_b$. We just need to check the two
    properties, which are both in \Ptime.
  \item    For $H^{2b}_{r,r}$, a graph $G$ maps to $H^{2b}_{r,r}$ if and only if it has no cycles with an odd number of blue edges (see Lemma~\ref{lemma:duality-H_{r,r}^{2b}}, proved in~\cite{BBM05}). 
   This property is preserved under the switching operation. Thus, switching the graph does not impact the nature of the instance. It is thus in \Ptime (we can test with $k=0$) since \textsc{$H^{2b}_{r,r}$-Colouring} is in \Ptime~\cite{BBM05,BDHQ05}.
  \end{itemize}

We now consider the \NP-complete cases. For every $H$,
\SWITCH{$H$} clearly lies in \NP.
\NP-hardness follows from the above-stated Theorem~\ref{thm:dicho-signed-col} (proved in~\cite{SignedDicho,BS18}) in all but one case: indeed,  $H^{2b}_{r,b}$, $H^{2rb}_{r,b}$, $H^{2rb}_{r,-}$ and $H^{2rb}_{r,r}$ are their own switching cores and have at least three edges, thus when $H$ is one of these, \SWITCH{$H$} is \NP-complete (even with $k=|V(G)|$).

The last case is $H^{2b}_{r,-}$. We give a
  reduction from \textsc{Vertex Cover} to
  \SWITCH{$H^{2b}_{r,-}$}. Given instance $G$ of \textsc{Vertex
    Cover}, we construct an all-red copy $G'$ of $G$, and we attach to
  each vertex $v$ of $G$ a blue edge $vv'$, with a red loop on $v'$ (see
  Figure~\ref{fig:reduc-SWITCH-H^{2b}_{r,-}}).

  \begin{figure}[!htpb]
    \centering
    \scalebox{0.8}{\begin{tikzpicture}[every loop/.style={},scale=1.2]
      \node[blackvertex,label={270:$u$}] (u) at (2,0) {};
      \node[blackvertex,label={270:$v$}] (v) at (3,0) {};
      \node[blackvertex,label={270:$w$}] (w) at (4,0) {};
      \node[blackvertex,label={0:$u'$}] (u') at (2,0.7) {};
      \node[blackvertex,label={0:$v'$}] (v') at (3,0.7) {};
      \node[blackvertex,label={0:$w'$}] (w') at (4,0.7) {};

      \draw[thick,blue] (u') -- (u);
      \draw[thick,blue] (v') -- (v);
      \draw[thick,blue] (w') -- (w);
      \path[thick,red,dashed] (v')   edge[out=60,in=120,loop, min distance=5mm] node  {} (v');
      \path[thick,red,dashed] (u')   edge[out=60,in=120,loop, min distance=5mm] node  {} (u');
      \path[thick,red,dashed] (w')   edge[out=60,in=120,loop, min distance=5mm] node  {} (w');
      
      \draw[thick,dashed,red] (u) -- (v);
      \draw[thick,dashed,red] (u) to[bend right=35] (w);
      
      \path (0,0) node {$G$};
      \path (1,0) node {\ldots};
      
      \draw[line width=1.1pt] (2,0) ellipse (3.2cm and 0.5cm);  
    \end{tikzpicture}}
    \caption{Reduction from \textsc{Vertex Cover} to \SWITCH{$H^{2b}_{r,-}$}.}\label{fig:reduc-SWITCH-H^{2b}_{r,-}}
  \end{figure}

  Denote by $x$ the vertex of $H^{2b}_{r,-}$ with a loop, and by $y$
  the other one. Assume that $G$ has a vertex cover $C$ of size at
  most $k$. Denote by $G''$ the graph obtained from $G'$ by switching
  at the vertices of $C$.  We map every vertex $v'$ to $x$, every
  vertex of $C$ to $x$ and the remaining ones to $y$. Since $C$ is a
  vertex cover, each red edge of $G''$ is either a loop on some vertex
  $v'$, an edge $vv'$ with $v\in C$ or an edge $uv$ with $u,v\in
  C$. In each case, both endpoints are mapped on $x$. The blue edges
  of $G''$ are then either $vv'$ with $v\notin C$ or $uv$ with
  $u\in C$ and $v\notin C$. In both cases, the two endpoints are
  mapped to different vertices of $H^{2b}_{r,-}$; thus,
  $G''\ecto H^{2b}_{r,-}$.

  Conversely, assume that we can switch $G'$ at vertices from a set
  $S$ such that the resulting graph $G''$ maps to $H^{2b}_{r,-}$. Let
  $C$ be the set of vertices $v$ of $G$ such that $v$ or $v'$ lies in
  $S$. Note that $C$ has size at most $|S|$. We claim that $C$ is a
  vertex cover of $G$. Assume that there is an edge $uv$ in $G$ with
  $u,v\notin C$. By construction, $u,u',v,v'\notin S$, so $uu',vv'$
  are blue in $G''$, and $uv$ is red. Thus, $u,v$ have to be mapped to
  $x$, and $u',v'$ to $y$, a contradiction since $u'$ has a incident
  red loop in $G''$. Therefore $C$ is a vertex cover of $G$.
%
\end{proof}

\section{Parameterized complexity results}\label{sec:param}



\subsection{\VDEL{$H$} and \EDEL{$H$}}\label{sec:FPT-VDEL-EDEL}

For many edge-coloured graphs $H$ of order at most~$2$, we can show that \VDEL{$H$} and \EDEL{$H$} are FPT by giving ad-hoc reductions to \textsc{Vertex Cover}, \textsc{Odd Cycle Transversal} or a combination of both. However, a more powerful method is to generalise a technique from~\cite{BDHQ05} used to prove that \textsc{$H$-Colouring} is in \Ptime by reduction to \textsc{$2$-Sat} (see also~\cite{Bthesis}):

\begin{theorem}[Brewster et al. \cite{BDHQ05}]\label{thm:order2-poly-2SAT-original}
  Let $H$ be an edge-coloured graph of order at most~$2$. Then, for each instance $G$ of \textsc{$H$-Colouring}, there exists a polynomially computable \textsc{$2$-Sat} formula $F(G)$ that is satisfiable if and only if $G\ecto H$. Thus, \textsc{$H$-Colouring} is in \Ptime.
\end{theorem}
\begin{proof}[Proof (sketch)]
   The formula $F(G)$ from Theorem~\ref{thm:order2-poly-2SAT-original} contains a variable $x_v$ for each vertex $v$ of $G$, and for each edge $uv$, a set of clauses that depends on $H$, as described in Table~\ref{table:2SAT-order2} (reproduced from~\cite{BDHQ05}).
The idea is to see the two vertices of $H$ as ``true'' ($1$) and ``false'' ($0$), and for each edge $uv$ of a certain colour, to express the possible valid assignments of $x_u$ and $x_v$ based on the edges of that colour that are present in $H$. For example, if $H$ has, for colour $i$, a loop at vertex $0$ and an edge $01$, but no other edge of colour $i$, for each edge $uv$ of $G$ of colour $i$, we add the clause $(\overline{x_u}+\overline{x_v})$ to $F(G)$, indeed the constraint for edge $uv$ is satisfied if at least one of $u,v$ is mapped to $0$.
\end{proof}

\begin{table}[htpb!]
 \centering
   \scalebox{1}{\begin{tabular}{p{2.5cm}l}
       $E_i(H)$ & Clause\\[2mm]
       \hline\\
       $\emptyset$ & $(x_u)(\overline{x_u})$\\[2mm]
       $\{00\}$ & $(\overline{x_u})(\overline{x_v})$\\[2mm]
       $\{01\}$ & $(x_u+x_v)(\overline{x_u}+\overline{x_v})$\\[2mm]
       $\{11\}$ & $(x_u)(x_v)$\\[2mm]
       $\{00,01\}$ & $(\overline{x_u}+\overline{x_v})$ \\[2mm]
       $\{01,11\}$ & $(x_u+x_v)$ \\[2mm]
       $\{00,11\}$ & $(x_u+\overline{x_v})(\overline{x_u}+x_v)$ \\[2mm]
       $\{00,01,11\}$ & $(x_u+\overline{x_u})$      
     \end{tabular}}
 \caption{Clauses appearing in the \textsc{$2$-Sat} formula $F(G)$ of Theorem~\ref{thm:order2-poly-2SAT-original} proved in~\cite{BDHQ05}, for each edge $uv$ of $G$ coloured $i$. The clauses depend on the edge set of $H$ in colour $i$, described in the rows (where $V(H)=\{0,1\}$).} \label{table:2SAT-order2}
\end{table}

We will show how to generalise this idea to \VDEL{$H$} and \EDEL{$H$}. We will need the following parameterized variant of \textsc{$2$-Sat}:

\problemparam{\textsc{Variable Deletion Almost $2$-Sat}}{A $2$-CNF Boolean formula $F$, an integer $k$.}{$k$.}{Is there a set of $k$ variables that can be deleted from $F$ (together with the clauses containing them) so that the resulting formula is satisfiable?}

\textsc{Variable Deletion Almost $2$-Sat} and another similar variant, \textsc{Clause Deletion Almost $2$-Sat} (where instead of $k$ variables, $k$ clauses may be deleted), are known to be FPT (see~\cite[Chapter 3.4]{BookParamAlgo} and~\cite{RO08}). We need to introduce a more general variant, that we call \textsc{Group Deletion Almost $2$-Sat}, defined as follows.

\problemparam{\textsc{Group Deletion Almost $2$-Sat}}{A $2$-CNF Boolean formula $F$, an integer $k$, and a partition of the clauses of $F$ into groups such that each group has a variable which is present in all of its clauses.}{$k$.}{Is there a set of $k$ groups of clauses that can be deleted from $F$ so that the resulting formula is satisfiable?}

By a generalisation of \cite[Exercise 3.21]{BookParamAlgo} for
\textsc{Clause Deletion Almost $2$-Sat}, we obtain the following
complexity result for \textsc{Group Deletion Almost $2$-Sat}. 

\begin{proposition}\label{prop:groupSAT}
\textsc{Group Deletion Almost $2$-Sat} is FPT.
\end{proposition}
\begin{proof}
We will reduce the problem \textsc{Group Deletion Almost $2$-Sat} to the problem \textsc{Variable Deletion Almost $2$-Sat}. 

Take an instance $\mathcal{G}$ of \textsc{Group Deletion Almost
  $2$-Sat} with groups $g_1,\dots,g_p$. We construct an instance
$\mathcal{V}$ of \textsc{Variable Deletion Almost $2$-Sat} as
follows. For $i\in [1,p]$, we replace each occurrence of variable $x$
in the clauses of group $g_i$ by a new variable $x_i$. Moreover, for each
variable $x$ and for each $i,j$, such that $1\leq i<j \leq p$, we add
the two clauses $(\overline{x_i} + x_j)$ and $(x_i + \overline{x_j})$
to $\mathcal{V}$ (i.e. $x_i = x_j$). The parameter for $\mathcal{V}$ remains $k$.

Suppose that $\mathcal{V}$ is a positive instance, i.e. that after removing
up to $k$ variables, the resulting set of clauses $\mathcal{V}'$ is
satisfied by a truth assignment $v$. For each removed variable $x_i$,
we remove the group of clauses $g_i$ in $\mathcal{G}$. Note that at
most $k$ groups are removed since we removed at most $k$ variables in
$\mathcal{V}$. We have to show that the new set of clauses
$\mathcal{G}'$ is satisfiable. 

Note that if $x_i$ and $x_j$ are not removed, then $v$ satisfies
$(\overline{x_i} + x_j)$ and $(x_i + \overline{x_j})$, which ensures
that $v(x_i)=v(x_j)$. Thus, defining the truth value of $x$ by the
value of $v(x_i)$ (for some non-removed $x_i$) is well-defined. Take a
clause $(x+y)$ of $\mathcal{G}'$, then $(x_i + y_i)$ is a satisfied
clause of $\mathcal{V}'$ for some $i\in [1,p]$. By definition of our
truth assignment, $(x+y)$ is satisfied, so $\mathcal{G'}$ is
satisfiable. Therefore, $\mathcal{G}$ is a positive instance.

Conversely, suppose that we can remove $k$ groups from $\mathcal{G}$
such that the resulting set of clauses $\mathcal{G}'$ is satisfied by
$v$. If we removed the group $g_i$ in the solution, then we remove
$x_i$ in $\mathcal{V}$ where $x_i$ is a variable of $g_i$ that appears
in each of its clauses. Such a variable exists by definition of
$\mathcal{G}$. This removes all the clauses corresponding to the
clauses of the group $g_i$ in $\mathcal{V}$. Thus, taking the truth
assignment that assigns to each $x_i$ the value $v(x)$ satisfies the
instance $\mathcal{V}$.
\end{proof}

We are now able to prove the following theorem.
 
\begin{theorem}\label{thm:VDEL-order2-FPT} \label{thm:EDEL-order2-FPT}
  For every edge-coloured graph $H$ of order at most~$2$, \VDEL{$H$} and \EDEL{$H$} are FPT.
\end{theorem}
\begin{proof}
  For an instance $G,k$ of \VDEL{$H$} or \EDEL{$H$}, we consider the formula $F(G)$ from Theorem~\ref{thm:order2-poly-2SAT-original} (see~Table~\ref{table:2SAT-order2}). In $F(G)$, to each vertex of $G$ corresponds a variable $x_v$. Deleting $v$ from $G$ when mapping $G$ to $H$ has the same effect as deleting $x_v$ when satisfying $F(G)$. Thus, this is an FPT reduction from \VDEL{$H$} to \textsc{Variable Deletion Almost $2$-Sat}.

Moreover, each edge $uv$ of $G$ corresponds to one or two clauses of $F(G)$. This naturally defines the groups of \textsc{Group Deletion Almost $2$-Sat} by grouping the clauses corresponding to the same edge. Removing an edge is equivalent to removing its corresponding group. To finish, we have to make sure that we can have one variable common to all the clauses of each group. This is the case in the reduction in~\cite{BDHQ05} for every case except when $E_i(H)$ (the set of
edges of colour $i$ in $H$) is just a loop. Assume without loss of generality that the loop is on vertex~$1$ (the other loop can be treated the same way). Suppose $uv$ has colour $i$ in $G$; then $uv$ must be mapped to
the loop on vertex~$1$. The original reduction added the clauses $(x_u)(x_v)$; we modify this part and add instead the clauses $(c + x_u)(c+x_v)(\overline{c})$ where $c$ is a new variable. This is now a valid and equivalent instance of \textsc{Group Deletion Almost $2$-Sat}, which is FPT by Proposition~\ref{prop:groupSAT}. 
\end{proof}



\subsection{\SWITCH{$H$}: FPT cases}

We now consider the parameterized complexity of \SWITCH{$H$}. By Theorem~\ref{thm:SWITCH-order2-poly/NPc}, there are five $2$-edge-coloured graphs $H$ of order at most~$2$ with \SWITCH{$H$} \NP-complete. We first show that two of them are FPT:

\begin{theorem}\label{thm:SWITCH-order2-FPT}
  \SWITCH{$H^{2b}_{r,b}$} and \SWITCH{$H^{2b}_{r,-}$} are FPT.
\end{theorem}
\begin{proof}
  The graph $H^{2b}_{r,b}$ has the finite duality property by~\cite{BBM05}, see Lemma~\ref{lemma:duality-H_{r,b}^{2b}}: $G \ecto H^{2b}_{r,b}$ if and only if $G$ does not contain a walk $abcd$ where $ab$ and $cd$ are red edges and $bc$ is a blue edge. This implies FPT time for \SWITCH{$H^{2b}_{r,b}$} by a simple bounded search tree algorithm (Proposition~\ref{prop:finite-duality-FPT}).
  
For the graph $H^{2b}_{r,-}$, as mentioned in Lemma~\ref{lemma:duality-H_{r,-}^{2b}}, the duality set $\mathcal F(H)$ discovered in~\cite{BBM05} is composed of walks of the form $RB^{2p-1}R$ (where $R$ is a red edge, $B$ a blue edge and $p\geq 1$ is an integer) and of closed walks with an odd number of blue edges. As seen before, if the graph $G$ has such a cycle then switching will not remove it, thus we can reject. 

If the graph has a $RB^{2p-1}R$ walk and is a positive instance, then we claim that we need to switch one of the four vertices incident with the red edges. Indeed, if we switch only at the vertices inside the blue walk (those not incident with one of the red edges) then the parity of the number of blue edges will not change and we will still have some maximal odd blue subwalk, the two edges next to the extremities being red. Thus we would still have a $RB^{2q-1}R$ path.

Thus, since we need to switch at one of these four vertices, we branch on this configuration using the classic bounded search tree technique. This is an FPT algorithm.
\end{proof}

\subsection{\SWITCH{$H$}: \W-hard cases}


The remaining cases, $H^{2rb}_{r,b}$, $H^{2rb}_{r,-}$ and $H^{2rb}_{r,r}$, yield \W-hard \SWITCH{$H$} problems, even for input graphs of large girth (the \emph{girth} of a graph is the smallest length of one of its cycles, and by the girth of an edge-coloured graph we mean the girth of its underlying uncoloured graph):

\begin{theorem}
  \label{thm:W1-general}
  Let $x\in\{r,b,-\}$. Then for any integer $g\geq 3$, the problem \SWITCH{$H_{r,x}^{2br}$} is $\W$-hard, even for graphs $G'$ with girth at least $g$ and that would map to $H_{r,x}^{2br}$ if the number of switchings was unbounded.
Under the same conditions, \SWITCH{$H_{r,x}^{2br}$} cannot be solved in time $f(k)|G|^{o(k)}$ for any computable function  $f$, assuming the ETH.
\end{theorem}

We will prove Theorem~\ref{thm:W1-general} by three reductions from \textsc{Multicoloured Independent Set}, which is \W-complete~\cite{PIETRZAK2003} and defined as follows.

\problemparam{\textsc{Multicoloured Independent Set}}{A graph $G$, an integer $k$ and a partition of $V(G)$ into $k$ sets $V_1$,\dots,$V_k$.}{$k$.}{Is there a set $S$ of exactly $k$ vertices of $G$, such that each $V_i$ contains exactly one element of $S$, that forms an independent set of $G$?}

Our three reductions (one for each possible choice of $x$) follow the
same pattern. In Section~\ref{sub:generic}, we describe this idea,
together with the required properties of the gadgets. In
Sections~\ref{sub:double-coton-tige},~\ref{sub:double-sucette}
and~\ref{sub:double-coton-tige-sale}, we show how to construct the
gadgets. Since the reduction preserves the parameter and is actually polynomial, the ETH-based lower bound follows from \cite{Chen2006}.

\subsubsection{Generic reduction}
\label{sub:generic}

Let $(G,k)$ be an instance of \textsc{Multicoloured Independent Set}, and denote
by $V_1,\ldots,V_k$ the partition of $G$. We begin by replacing each
$V_i$ by a \emph{partition gadget} $G_i$. This gadget must have
$|V_i|$ special vertices $x_j \in V_i$, in order to associate a vertex of $G_i$
to each vertex of $V_i$. Moreover, $G_i$ must satisfy the following:
\begin{itemize}
\item[$(P1)$] We do not have $G_i\ecto H_{r,x}^{2rb}$.
\item[$(P2)$] If we switch $G_i$ at exactly one vertex $v$, then the
  obtained graph maps to $H_{r,x}^{2rb}$ (without switching) if
  and only if $v$ is one of the special vertices of $G_i$.
\item[$(P3)$] $G_i$ has girth at least $g$.
\item[$(P4)$] $G_i$ has two \emph{reset vertices} $x$ and $y$ that are different from the $x_i$'s and such that $G_i$ switched at $x$ and $y$ maps to $H_{r,x}^{2rb}$ (without further switching).
\end{itemize}

Let $uv$ be an edge of $G$. Recall that $u$ and $v$ can be seen as
vertices of $G'$. We then add an \emph{edge gadget} $G_{uv}$ between $u$
and $v$. This gadget must satisfy the following:
\begin{itemize}
\item[$(E1)$] Let $H$ be the graph obtained from $G_{uv}$ by switching
  at a subset $S$ of $\{u,v\}$. Then, $H\ecto H_{r,x}^{2rb}$ if
  $S\neq\{u,v\}$.
\item[$(E2)$] Assume that $u \in V_i$ and $v \in V_j$ and let $H$ be
  the graph obtained from $G_{uv}\cup G_i\cup G_j$ by switching $u$
  and $v$. Then, we do \emph{not} have $H\ecto H_{r,x}^{2rb}$.
\item[$(E3)$] $G_e$ has girth at least $g$.
\item[$(E4)$] In $G_e$, $u$ and $v$ are at distance at least $g$.
\end{itemize}

Let $G'$ be the graph obtained from $G$ by replacing each $V_i$ by a
partition gadget $G_i$, and each edge $uv$ by an edge gadget $G_{uv}$
such that for every $u\in V_i$ and $v$ such that $uv$ is an edge, we
identify the special vertex $u$ in $G_i$ with the special vertex $u$
in $G_{uv}$. (Note in particular that every vertex of $G$ is present in
$G'$.)

We say that a set $S$ of vertices of $G$ is \emph{valid} if, when seen
in $G'$, it contains at most one special vertex in each edge
gadget. We need a last condition about $G'$:
\begin{itemize}
\item[$(SP)$] If, after switching a valid set in $G'$, the obtained
  graph does not map to $H_{r,x}^{2rb}$, then this is because a partition gadget or an edge gadget does not map to $H_{r,x}^{2rb}$ (that is, each minimal obstruction is entirely contained in an edge gadget or a partition gadget).
\end{itemize}
With this Property~$(SP)$, we can prove that $(G,k)\mapsto (G',k)$ is
a valid reduction. 
\begin{proposition}
  $(G',k)$ is a positive instance of \SWITCH{$H_{r,x}^{2rb}$} if and
  only if $(G,k)$ is a positive instance of \textsc{Multicoloured Independent Set}.
\end{proposition}
\begin{proof}
  Assume we can switch at most $k$ vertices of $G'$ such that the
  obtained graph maps to $H_{r,x}^{2rb}$. Let $S$ be the set of those
  vertices. We claim that $S$ is a valid set of $G'$. First note that,
  due to $(P1)$, $S$ must contain at least one vertex in each
  $V_i$. This enforces $|S|=k$, thus $S$ contains exactly one vertex
  $v_i$ in each $V_i$. By $(P2)$, each of these $v_i$ has to be one of the special vertices of $G_i$.
  This means that $S$ contains only vertices that are present
  in $G$.
    
  We claim that $S$ induces an independent set in $G$. Assume by contradiction that
  there is an edge $uv$ in $G$ with $u,v\in S$. Then, by
  construction, there is an edge gadget whose special vertices are $u$
  and $v$, such that the edge gadget and the two partition gadgets
  associated with $u$ and $v$ map to $H_{r,x}^{2rb}$ when we switch
  only at $u$ and $v$, contradicting $(E2)$. (Note that $S$
  does not contain any other vertex of the edge gadget nor any other
  vertex of the partition gadgets.) Therefore, $G$ has an
  independent set of size $k$ containing exactly one vertex in each
  set $V_i$.

  Conversely, assume that $G$ has an independent set $S$ intersecting
  each $V_i$ at one vertex. Then, we denote by $H$ the graph obtained
  by switching all vertices of $S$ in $G'$. By construction, this is a
  valid set, hence by $(SP)$ every obstruction for mapping to $H_{r,x}^{2rb}$ in $H$ is actually
  contained in some gadget. However, it cannot be contained in a
  partition gadget due to $(P2)$, nor in an edge gadget due to
  $(E1)$. Therefore, we have $H\ecto H_{r,x}^{2rb}$.
\end{proof}

Observe moreover that, due to $(P3)$, $(E3)$ and $(E4)$, $G'$ has
girth at least $g$. Moreover, let $S$ be the set of all reset vertices
of $G'$. Let $H$ be the $2$-coloured graph obtained by switching every
vertex of $S$. By $(P4)$, no partition gadget in $H$ contains an
obstruction. Furthermore, no edge gadget contains an obstruction by
$(E1)$. Therefore, using $(SP)$, we obtain that $H$ does not contain
any obstruction, hence $H\ecto H_{r,x}^{2rb}$. Thus to prove
Theorem~\ref{thm:W1-general} it suffices to construct the gadgets.


\subsubsection{Gadgets for $H_{r,r}^{2rb}$}
\label{sub:double-coton-tige}

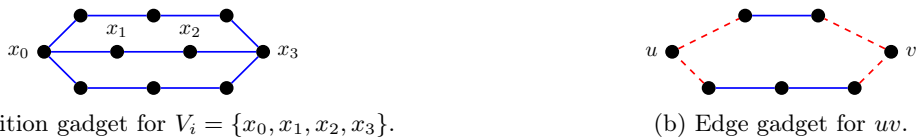
\begin{figure}[!htpb]
  \begin{subfigure}[t]{0.45\textwidth}
        \centering
      \scalebox{0.8}{\begin{tikzpicture}[every loop/.style={},scale=1.2]
      
      \node[blackvertex,label=left:$x_0$] (a0) at (0,0) {};
      \node[blackvertex,label=above:$x_1$] (a1) at (1,0) {};
      \node[blackvertex,label=above:$x_2$] (a2) at (2,0) {};
      \node[blackvertex,label=right:$x_3$] (a3) at (3,0) {};
      \node[blackvertex] (t1) at (0.5,0.5) {};
      \node[blackvertex,label=above:$r_1$] (t2) at (1.5,0.5) {};%
      \node[blackvertex] (t3) at (2.5,0.5) {};
      
      \node[blackvertex] (b1) at (0.5,-0.5) {};
      \node[blackvertex,label=above:$r_2$] (b2) at (1.5,-0.5) {};
      \node[blackvertex] (b3) at (2.5,-0.5) {};
      
      \draw[thick,blue] (a0) -- (t1) -- (t2) -- (t3) -- (a3) -- (b3) -- (b2) -- (b1) -- (a0);
      \draw[thick,blue] (a0) -- (a1) -- (a2) -- (a3);
	  
      \end{tikzpicture}}
      \caption{Partition gadget for $V_i =\{x_0,x_1,x_2,x_3\}$ with the two reset vertices $r_1$, $r_2$.}
          \label{fig:partitionrr}
    \end{subfigure}\hfill
  \begin{subfigure}[t]{0.45\textwidth}
        \centering
   \scalebox{0.8}{\begin{tikzpicture}[every loop/.style={},scale=1.2]
      
      \node[blackvertex,label=left:$u$] (a0) at (0,0) {};
      \node[blackvertex,label=right:$v$] (a3) at (3,0) {};
      
      \node[blackvertex] (t1) at (1,0.5) {};
      \node[blackvertex] (t3) at (2,0.5) {};
      
      \node[blackvertex] (b1) at (0.5,-0.5) {};
      \node[blackvertex] (b2) at (1.5,-0.5) {};
      \node[blackvertex] (b3) at (2.5,-0.5) {};
      
      \draw[thick,blue]  (t1) --  (t3);
      \draw[thick,blue] (b3) -- (b2) -- (b1) ;
      \draw[thick,red,dashed] (b1) -- (a0) -- (t1);
      \draw[thick,red,dashed] (b3) -- (a3) -- (t3);
   \end{tikzpicture}}
   \caption{Edge gadget for $uv$.}
       \label{fig:edgerr}
    \end{subfigure}
    \caption{Partition and edge gadgets in the $H_{r,r}^{2rb}$-reduction when $g = 3$.}

  \end{figure}
   
  We now describe the gadgets for \SWITCH{$H_{r,r}^{2rb}$}. As mentioned in Lemma~\ref{lemma:duality-H_{r,r}^{2rb}}, for every $2$-edge-coloured graph
$G$, we have $G\ecto H_{r,r}^{2rb}$ if and only if it does not
contain an all-blue odd cycle.

The partition gadget $G_i$ is an all-blue cycle of length $2g$ if
$g$ and $|V_i|$ have the same parity (resp. $2g +2$ is they do not have
the same parity) with a chord of order $|V_i|$ between two antipodal
vertices. The special vertices are those on the
chord (see Figure~\ref{fig:partitionrr}).
The reset vertices are defined as any two vertices on the initial cycle, one on each side of the chord.

Property $(P3)$ directly follows from the construction. Moreover,
since $G_i$ contains an all-blue odd cycle, we have $(P1)$. If we
switch $G_i$ at exactly one vertex, then either this vertex is a
special vertex and the obtained graph does not have any all-blue odd
cycle (and thus maps to $H_{r,r}^{2rb}$), or it is not a special
vertex and there is still an all-blue odd cycle. Therefore,
property $(P2)$ also holds.

Finally, if we switch at the two reset vertices, then there is no more all-blue odd cycle, thus $(P4)$ also holds.

We now consider the edge gadget. It is formed by an all-blue odd
cycle of length $2g+1$ where two vertices $u,v$ at distance $g$ have
been switched (see Figure~\ref{fig:edgerr}). These vertices are the
special vertices of the gadget. By construction, properties $(E3)$ and
$(E4)$ hold. Moreover, consider a set $S\subset\{u,v\}$. The only way
for switching the vertices of $S$ to yield a graph containing an
all-blue odd cycle is to switch both $u$ and $v$. This proves
$(E1)$. If we switch at both special vertices then
we do not have $G_{uv} \ecto H^{2rb}_{r,r}$, which implies $(E2)$.

It remains to prove Property~$(SP)$. Let $S$ be a valid set, and
let $H$ be the graph obtained from $G'$ when switching all vertices of
$S$. Assume that $H$ contains an all-blue odd cycle. Since $S$ is
valid set, at most one vertex has been switched in each edge
gadget. Therefore, no all-blue odd cycle of $H$ can contain an
edge from an edge gadget. It is thus contained in some partition
gadget, ensuring that $(SP)$ holds.

\subsubsection{Gadgets for $H_{r,-}^{2rb}$}
\label{sub:double-sucette}

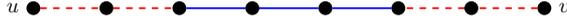
\begin{figure}[!htpb]
    \centering
    \scalebox{0.8}{\begin{tikzpicture}[every loop/.style={},scale=1.2]
      
      \node[blackvertex,label=left:$u$] (a0) at (-1,0) {};
      \node[blackvertex,label=right:$v$] (a3) at (6,0) {};
      
      \node[blackvertex] (b0) at (0,0) {};
      \node[blackvertex] (b1) at (1,0) {};
      \node[blackvertex] (b2) at (2,0) {};
      \node[blackvertex] (b3) at (3,0) {};
      \node[blackvertex] (b4) at (4,0) {};
      \node[blackvertex] (b5) at (5,0) {};
      
      \draw[thick, blue] (b1) -- (b2) -- (b3) -- (b4);
      \draw[thick, red, dashed] (a0) --(b0) --  (b1);
      \draw[thick, red, dashed] (a3) -- (b5) -- (b4);


    \end{tikzpicture}}
    \caption{The edge gadget for $uv$ in the $H_{r,-}^{2rb}$-reduction when $g = 6$.}
    \label{fig:edger-}
  \end{figure}

  We now describe the gadgets for \SWITCH{$H_{r,-}^{2rb}$}. As mentioned in lemma~\ref{lemma:duality-H_{r,-}^{2rb}},
  for every $2$-edge-coloured graph $G$, we have $G\ecto H_{r,-}^{2rb}$ if and only if
  it does not contain a \emph{bad} walk, i.e. an odd figure eight $v_0$, $v_1$,  \dots, $v_{2j}$, $v_0$, $v_{2j+2}$, \dots, $v_{2p-1}$, $v_0$ such
  that all edges $v_{2i}v_{2i+1}$ are blue~\cite{BBM05}.

The partition gadget $G_i$ is the same as in the previous case (see
Figure~\ref{fig:partitionrr}).

The edge gadget is an odd path of length at least $g$, whose edges are
all blue except for the two first and two last ones (see
Figure~\ref{fig:edger-}).

Since the partition gadget $G_i$ is the same as for $H_{r,r}^{2rb}$, Property $(P3)$ still
holds. Moreover, since all-blue odd-cycles still are obstructions,
we have $(P1)$.

Observe that if a graph $H$ contains an obstruction, then so does its
subgraph obtained by removing recursively its leaves. Note that
switching exactly one vertex $v$ in $G_i$ makes its neighborhood
all-red. Therefore, $v$ cannot be contained in a bad walk
anymore. In this case, the obstruction is contained in a possibly
empty cycle $C_v$ (obtained by removing from $G_i$ the vertex $v$ and
the leaves of $G_i$ recursively).

If we switch $G_i$ at exactly one vertex, then either this vertex is a
special vertex and $C_v$ is empty or an all-blue even cycle (and
thus maps to $H_{r,-}^{2rb}$), or it is not a special vertex and $C_v$
is still an all-blue odd cycle. Therefore, property $(P2)$ also
holds.

Finally, if we switch at the two reset vertices $u,v$, then $G_i\setminus\{u,v\}$ is a tree, thus $G_i$ does not contain any obstruction, hence $(P4)$ also holds.

By construction, properties $(E3)$ and
$(E4)$ hold. Moreover, observe that the edge gadget does not contain a
bad walk since it is a path. Thus $(E1)$ holds. If $H$ is the graph
defined in property $(E2)$ then there is a bad walk starting from $u$,
then turning around one odd cycle in the partition gadget containing
$u$, crossing the edge gadget to $v$, taking a similar turn around an
odd cycle of the partition gadget containing $v$ and then going back
to $u$ by the edge gadget. So $(E2)$ holds. 

It remains to prove $(SP)$. Let $S$ be a valid set, and
$H$ be the graph obtained from $G'$ by switching $S$. Observe that no
bad walk contains to consecutive red edges. Moreover, in $H$,
every edge gadget contains two such edges (since its two endpoints
cannot be both in $S$). Therefore, no bad walk crosses an edge gadget
$G_{uv}$, which implies that no bad walk contains edges in
$G_{uv}$. Hence, every bad walk is contained in some partition gadget,
thus ensuring that $(SP)$ holds.

\subsubsection{Gadgets for $H_{r,b}^{2rb}$}
\label{sub:double-coton-tige-sale}

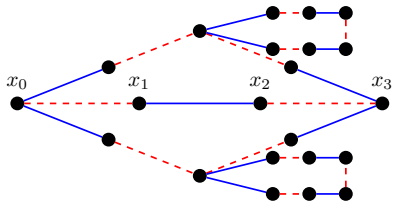
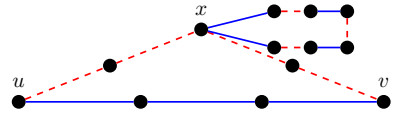
\begin{figure}[!htpb]
  \begin{subfigure}[t]{0.45\textwidth}
    \centering
    \scalebox{0.8}{\begin{tikzpicture}[every loop/.style={},scale=1.2]
\tikzstyle{n} = [blackvertex];
\tikzstyle{b} = [thick, color=blue]; 
\tikzstyle{r} = [thick, color=red, dashed];
\def\ratio{100} 

\node[n, label=above:$x_0$] (v1) at (-250.00/\ratio,50.00/\ratio) {};
\node[n, label=above:$x_3$] (v3) at (250.00/\ratio,50.00/\ratio) {};

\node[n, label=above:$x_1$] (v48) at (-83.00/\ratio,50.00/\ratio) {};
\node[n, label=above:$x_2$] (v49) at (83.00/\ratio,50.00/\ratio) {};
\node[n] (v50) at (-125.00/\ratio,100.00/\ratio) {};
\node[n,label=above:$r_1$] (v51) at (0.00/\ratio,150.00/\ratio) {};
\node[n] (v52) at (125.00/\ratio,100.00/\ratio) {};
\node[n] (v53) at (100/\ratio,175/\ratio) {};
\node[n] (v54) at (150/\ratio,175/\ratio) {};
\node[n] (v55) at (200/\ratio,175/\ratio) {};
\node[n] (v56) at (100/\ratio,125/\ratio) {};
\node[n] (v57) at (150/\ratio,125/\ratio) {};
\node[n] (v58) at (200/\ratio,125/\ratio) {};

\node[n] (b50) at (-125.00/\ratio,00.00/\ratio) {};
\node[n,label=above:$r_2$] (b51) at (0.00/\ratio,-50.00/\ratio) {};
\node[n] (b52) at (125.00/\ratio,00.00/\ratio) {};
\node[n] (b53) at (100/\ratio,-75/\ratio) {};
\node[n] (b54) at (150/\ratio,-75/\ratio) {};
\node[n] (b55) at (200/\ratio,-75/\ratio) {};
\node[n] (b56) at (100/\ratio,-25/\ratio) {};
\node[n] (b57) at (150/\ratio,-25/\ratio) {};
\node[n] (b58) at (200/\ratio,-25/\ratio) {};

\draw (v1) edge[b] (v50);
\draw (v50) edge[r] (v51);
\draw (v51) edge[r] (v52);
\draw (v52) edge[b] (v3);
\draw (v53) edge[r] (v54);
\draw (v56) edge[r] (v57);
\draw (v55) edge[r] (v58);
\draw (v53) edge[b] (v51);
\draw (v51) edge[b] (v56);
\draw (v54) edge[b] (v55);
\draw (v57) edge[b] (v58);
\draw (v1) edge[r] (v48);
\draw (v48) edge[b] (v49);
\draw (v49) edge[r] (v3);

\draw (v1) edge[b] (b50);
\draw (b50) edge[r] (b51);
\draw (b51) edge[r] (b52);
\draw (b52) edge[b] (v3);
\draw (b53) edge[r] (b54);
\draw (b56) edge[r] (b57);
\draw (b55) edge[r] (b58);
\draw (b53) edge[b] (b51);
\draw (b51) edge[b] (b56);
\draw (b54) edge[b] (b55);
\draw (b57) edge[b] (b58);
	  
    \end{tikzpicture}}
    \caption{Partition gadget for $V_i =\{x_0,x_1,x_2,x_3\}$, with the two reset vertices $r_1$, $r_2$.}
        \label{fig:partitionrb}
    \end{subfigure}\hfill
  \begin{subfigure}[t]{0.45\textwidth}
    \scalebox{0.8}{\begin{tikzpicture}[every loop/.style={},scale=1.2]
      \tikzstyle{n} = [blackvertex];
\tikzstyle{b} = [thick, color=blue]; 
\tikzstyle{r} = [thick, color=red, dashed];
\def\ratio{100} 

\node[n, label=above:$u$] (v1) at (-250.00/\ratio,50.00/\ratio) {};

\node[n, label=above:$v$] (v3) at (250.00/\ratio,50.00/\ratio) {};

\node[n] (v48) at (-83.00/\ratio,50.00/\ratio) {};
\node[n] (v49) at (83.00/\ratio,50.00/\ratio) {};
\node[n] (v50) at (-125.00/\ratio,100.00/\ratio) {};
\node[n, label=above:$x$] (v51) at (0.00/\ratio,150.00/\ratio) {};
\node[n] (v52) at (125.00/\ratio,100.00/\ratio) {};
\node[n] (v53) at (100/\ratio,175/\ratio) {};
\node[n] (v54) at (150/\ratio,175/\ratio) {};
\node[n] (v55) at (200/\ratio,175/\ratio) {};
\node[n] (v56) at (100/\ratio,125/\ratio) {};
\node[n] (v57) at (150/\ratio,125/\ratio) {};
\node[n] (v58) at (200/\ratio,125/\ratio) {};

\draw (v1) edge[r] (v50);
\draw (v50) edge[r] (v51);
\draw (v51) edge[r] (v52);
\draw (v52) edge[r] (v3);
\draw (v53) edge[r] (v54);
\draw (v56) edge[r] (v57);
\draw (v55) edge[r] (v58);
\draw (v53) edge[b] (v51);
\draw (v51) edge[b] (v56);
\draw (v54) edge[b] (v55);
\draw (v57) edge[b] (v58);
\draw (v1) edge[b] (v48);
\draw (v48) edge[b] (v49);
\draw (v49) edge[b] (v3);
    \end{tikzpicture}}
    \caption{Edge gadget for $uv$. The vertex $x$ is where the two alternating cycles were identified.}
    \label{fig:edgerb}
  \end{subfigure}
\caption{Partition and edge gadgets in the $H_{r,b}^{2rb}$-reduction when $g = 3$.}
  \end{figure}

We now describe the gadgets for \SWITCH{$H_{r,b}^{2rb}$}. As mentioned in Lemma~\ref{lemma:duality-H_{r,b}^{2rb}}, for every $2$-edge-coloured graph
$G$, we have $G\ecto H_{r,b}^{2rb}$ if and only if it does not
contain \emph{alternating odd figure eight}, that is, an alternating
  closed walk $v_0$, $v_1$,  \dots, $v_{2j}$, $v_0$, $v_{2j+2}$, \dots, $v_{2p-1}$, $v_0$ for some integers $j$ and $p$~\cite{BBM05}.

The partition gadget $G_i$ is defined by gluing two obstructions with
large girth along a path of length $|V_i|$ (see
Figure~\ref{fig:partitionrb}). More precisely, consider an alternating
odd cycle $C$ of size $|V_i|+g$ (or $|V_i|+g+1$). Note that $C$
contains a vertex $u$ adjacent to two red edges. We attach an
alternating odd cycle $C'$ of length $g$ (or $g+1$) to $u$, such that
the edges of $C'$ adjacent to $u$ are blue. To obtain $G_i$, we
take two copies of this obstruction, and glue their respective largest
cycle along a path of length $|V_i|$. The vertices of this path are
the special vertices of $G_i$, and the two copies of $u$ are the reset vertices of $G_i$.

The edge gadget is formed by identifying the vertices with
monochromatic neighbourhood of two alternating odd cycles of length
$2g+1$, in such a way that the common vertex has two blue edges in
one cycle and two red edges in the other one. To obtain the edge
gadget, we switch this graph at two vertices $u,v$ in the same cycle,
at distance $g$ from each other (see Figure \ref{fig:edgerb}).

Observe that $G_i$ has girth at least $g$, hence Property $(P3)$
holds. Moreover, by construction, $G_i$ contains an obstruction, hence
$(P1)$ holds. Note that there are exactly two (minimal) obstructions
in $G_i$, the ones used to construct it. Therefore, if we switch $G_i$
at a non-distinguished vertex, one of the these obstructions is
unchanged, and the obtained graph does not map to
$H_{r,b}^{2rb}$. Conversely, assume that we switch $G_i$ at a
distinguished vertex $u$ and there remains an obstruction. Note that
all the paths of length two starting from $u$ are now monochromatic,
hence no alternating odd figure eight can go through $u$. This implies that every alternating odd figure eight
in this graph does not use the internal vertices of the chord. When
removing these vertices from $G$, the former endpoints of the chord
have monochromatic neighborhood, hence they cannot be contained in an alternating odd figure eight. Removing the whole chord and (recursively) the leaves of
$G_i$ gives two disjoint alternating odd cycles, which do not contain
any alternating odd figure eight. Thus we have $(P2)$.

Finally, if we switch the two reset vertices of $G_i$, all the paths of length~$2$ starting at these vertices are monochromatic, hence no alternating odd figure eight goes through them. Removing the reset vertices, and recursively the obtained leaves gives the empty graph. Therefore, there is no alternating odd figure eight in $G_i$, it thus maps to $H_{r,b}^{2rb}$, and $(P4)$ holds.

The construction of the edge gadget ensures that $(E3)$ and $(E4)$
are satisfied. Moreover, if we switch at $u$ and $v$, we obtain an
obstruction, ensuring that $(E2)$ holds. Finally, let $H$ be the graph
obtained from $G_{uv}$ by possibly switching $v$. Then every path of
length two starting at $u$ is monochromatic, hence no alternating odd figure eight in $H$
contains $u$. Removing $u$ and leaves of $H$ yields an alternating odd
cycle, which does not contain any alternating odd figure eight. Therefore, $H$ maps to
$H_{r,b}^{2rb}$, and by exchanging $u$ with $v$, we obtain $(E1)$.

It remains to prove $(SP)$. Let $S$ be a valid set and
$H$ be the graph obtained from $G'$ by switching at every vertex of
$S$. Consider an alternating odd figure eight containing an edge from an edge gadget and an
edge from a partition gadget. This walk goes through a vertex
$u\in V_i$ such that the edge before $u$ in the walk lies in $G_i$ and
the other one lies in some $G_{uv}$. If $u\in S$, the paths of length
$2$ starting from $u$ in $G_i$ are monochromatic. Conversely, if
$v\notin S$, the paths of length $2$ starting at $u$ in $G_{uv}$ are
monochromatic. In both cases we reach a contradiction with the
existence of an alternating odd figure eight going through $v$. Therefore, every alternating odd figure eight
of $H$ is contained either in an edge gadget or in a partition gadget.

\section{Conclusion and perspectives}\label{sec:conclu}

We have introduced \VDEL{$H$}, \EDEL{$H$} and \SWITCH{$H$} and characterised their complexity for some small edge-coloured graphs $H$. The full complexity landscape still needs to be determined. We have fully classified the classic complexity of \VDEL{$H$} problems. It remains to do the same for \EDEL{$H$} and \SWITCH{$H$}.

We proved that both \VDEL{$H$} and \EDEL{$H$} are FPT when $H$ has order at most~$2$. However, if $H$ has order~$3$, for example if $H$ is a monochromatic triangle, we obtain \textsc{$3$-Colouring}, which is not even in \XP. \SWITCH{$H$} seems particularly interesting, since we obtained an FPT/\W-hard dichotomy when $H$ has order at most~$2$ (in which case the problem is always in \XP). But again for some $H$ of order~$3$, \SWITCH{$H$} is not in \XP. It would be very interesting to obtain FPT/\W/\XP trichotomies for \VDEL{$H$}, \EDEL{$H$} and \SWITCH{$H$}, as least for some interesting classes of targets $H$ such as, for example, trees or cycles.

One may also study restricted classes of inputs, such as planar graphs (studied in the context of switching homomorphisms in~\cite{planar}). For example, do the \W-hard cases of \SWITCH{$H$} become FPT (or even polynomial) when the input is planar?

Another variation that seems of interest, recently studied in~\cite{lists} for signed graphs and $2$-edge-coloured graphs, is the one when lists are involved (the input is given with a list function that assigns to each vertex, an allowed set of vertices from the target graph $H$). What are the complexities of list versions of \VDEL{$H$}, \EDEL{$H$} and \SWITCH{$H$} for $2$-edge-coloured graphs $H$?

One could also study VD or ED versions of \textsc{Signed $H$-Colouring}. For example, in~\cite{HBN10}, the authors proved that \textsc{ED Balanced Subgraph}, the problem of deciding whether a given signed graph becomes balanced after $k$ edge-deletions, is FPT for parameter $k$. (Note that the minimum number of edge/vertex-deletions required to make a signed graph balanced is studied under the name of \emph{frustration index/number} of that signed graph~\cite{Zpb}.) This problem is equivalent to the ED-version of \textsc{Signed $H_{b,b}^{2r}$-Colouring}.

Finally, we note that it could be interesting to study analogues of \VDEL{$H$} and \EDEL{$H$} for arbitrary fixed-template CSP problems, not just when $H$ is an edge-colored graph. 
To the best of our knowledge, this has not been done.

\end{document}